\documentclass[a4, 11pt]{scrartcl}
\usepackage{fullpage}

\usepackage[noadjust]{cite}
\usepackage{amssymb}
\usepackage{sectsty}%,titling}
\usepackage[utf8]{inputenc}
\usepackage[english]{babel}
\usepackage[svgnames]{xcolor}
\usepackage{graphicx}
\usepackage{amsmath}
\usepackage{amsthm}
\usepackage{csquotes}
\usepackage{hyperref}
\usepackage{cleveref}
\usepackage[draft]{todonotes}
\usepackage[ruled,vlined]{algorithm2e}
\usepackage{mathtools}
\usepackage{amsfonts}
\usepackage{enumerate}
\usepackage{comment}
\usepackage{thmtools}
\usepackage{thm-restate}
\usepackage{multicol}
\usepackage[affil-sl]{authblk}
\usepackage{multirow}

\usetikzlibrary{calc, decorations.pathreplacing, calligraphy, patterns.meta}

\newtheorem{theorem}{Theorem}[section]
\newtheorem{corollary}[theorem]{Corollary}
\newtheorem{lemma}[theorem]{Lemma}
\newtheorem{fact}[theorem]{Fact}

\newtheorem{invariant}[theorem]{Invariant}
\theoremstyle{definition}

\newcommand{\mcO}{\mathcal{O}}
\newcommand{\OPT}{\mathrm{OPT}}
\newcommand{\OPTS}{\ensuremath{\operatorname{OPT_S}}}
\newcommand{\OPTD}{\ensuremath{\operatorname{OPT}}}
\newcommand{\ALG}{\ensuremath{\operatorname{ALG}}}

\newcommand{\SODA}{\ensuremath{\operatorname{\operatorname{\texttt{LEARN}}}}}
\newcommand{\FLOW}{\operatorname{\operatorname{\texttt{FLOW}}}}

\def\cost{\operatorname{cost}}

\def\favg{\operatorname{free-avg}}
\def\savg{\operatorname{size-avg}}
\def\ffv{\operatorname{ffv}}
\def\charge{\operatorname{charge}}
\def\avg{\operatorname{charge-avg}}

\long\def\ruslan #1{{\color{blue}\bfseries\mathversion{bold}[#1 -- Ruslan]\mathversion{normal}}}

\title{Tight Bounds for Online Balanced Partitioning in the Generalized Learning Model\footnote{Research supported by the German Research Foundation (DFG), grants 47002938 (FlexNet) and SPP 2378 (ReNO).}}

\date{}

\author[1]{Harald R\"acke}
\author[2]{Stefan Schmid}
\author[1]{Ruslan Zabrodin}
\affil[1]{Technical University of Munich}
\affil[2]{Technical University of Berlin}

\begin{document}
\maketitle

\begin{abstract}
Resource allocation in distributed and networked systems such as the Cloud is 
becoming increasingly flexible, allowing these systems to dynamically
adjust toward the workloads they serve, in a demand-aware manner.

Online balanced partitioning is a fundamental optimization problem
underlying such self-adjusting systems.
We are given a set of $\ell$ servers. On each server we can schedule up to $k$ processes simultaneously (the \emph{capacity}); overall there are $n=k\cdot\ell$ processes. 
The demand is described as a sequence of requests 
$\sigma_t=\{p_i, p_{j}\}$, which means that the two processes  $p_i,p_{j} \in P$ communicate.
Whenever an algorithm learns about a new request, it is
allowed to move processes from one server to another, which however costs
1 unit per process move. If the processes (the endpoints of the request) are on different
servers, it further incurs a communication cost of 1 unit for this request.
The objective is to minimize the competitive ratio: the cost of serving such a request sequence compared to the cost incurred by an optimal offline algorithm. 

Henzinger et al. (at SIGMETRICS’2019) introduced a   learning variant of this problem where the cost of an online algorithm is compared to the cost of a static offline algorithm that does not perform any communication, but which simply learns the communication graph and keeps the discovered connected components together. This problem variant was recently also studied at SODA'2021.

In this paper, we consider a more general learning model (i.e., stronger adversary), where the offline algorithm is not restricted 
to keep connected components together.   
%(footnote then also online algorithm needs to keep it together both deterministic and randomized.)
Our main contribution are tight bounds for this problem.
In particular, we present two deterministic online algorithms: (1)~an online algorithm with competitive ratio $\mcO(\max(\sqrt{k\ell \log k}, \ell \log k))$ and augmentation $1+\epsilon$; (2)~an online algorithm with competitive ratio $\mcO(\sqrt{k})$ and augmentation $2+\epsilon$.
We further present lower bounds showing optimality of these bounds.

%discuss in related work: in Sigmetrics/SODA darf online nicht splitten macht
%aber keinen Unterschied für deterministischen aber für randomisierten könnte es Unterschied machen. bei unterer Schranke von competitive ratio ev. nicht mehr von SODA. 

%It finds
%applications in the context of resource allocation in the cloud
%and for optimizing distributed data structures such as union–
%find data structures.
\end{abstract}

\section{Introduction}
\label{sec:intro}

The performance of many distributed applications (e.g., deep learning models such as GPT-4) critically
depends on the performance of the underlying communication networks (e.g., during distributed training) \cite{arzani2023rethinking,mogul2012we}.
Especially large flows (also known as elephant flows) may consume significant network resources if communicated across multiple hops;  resources which would otherwise be available for additional flows~\cite{mellette2017rotornet,griner2021cerberus}.
A particularly innovative approach to improve the communication efficiency is to adjust the network resources dynamically (e.g., using virtualization), in a demand-aware manner: by moving frequently communicating vertices (e.g., processes or virtual machines) topologically closer (e.g., collocating them on the same server), transmissions over the network can be reduced.

When and how to collocate vertices however is an algorithmically challenging problem, as it introduces a tradeoff: as moving a vertex to a different server comes with overheads, it should not be performed too frequently and only when the moving cost can be amortized by the more efficient communication later. Devising good migration strategies is particularly difficult in the realm of online algorithms and competitive analysis, where the demand is not known ahead of time. 

A fundamental algorithmic problem underlying such self-adjusting infrastructures is known as 
online balanced (graph) partitioning, which has recently been studied intensively \cite{soda21,apocs21repartition,sidma19,infocom21repartitioning,disc16,computing18,sigmetrics19learn,netys17learn}, see also the recent SIGACT News article on the problem \cite{sigact}. 
In a nutshell (details will follow), in this model, we are given a set of $\ell$ servers, each of capacity $k$, which means that we can schedule up to $k$ processes simultaneously on each server. There are $n=k\cdot\ell$ processes in total. 
We need to serve a sequence of communication requests 
$\sigma_t=\{p_i, p_{j}\}$ with $p_i,p_{j} \in P$ between processes. The sequence is revealed one-by-one to an online algorithm. Upon a new request, the  algorithm is allowed to move processes from one server to another, which costs
1 unit per process move. It then needs to serve the request: if the processes (the endpoints of the request) are on different
servers, it further incurs a communication cost of 1 unit for this request.
The objective is to minimize the competitive ratio: the cost of serving such a request sequence compared to the cost incurred by an optimal offline algorithm. 

In this paper we are interested in the learning variant of this problem, initially introduced by Henzinger et al. \cite{sigmetrics19learn} 
at SIGMETRICS’2019 and later also studied at SODA'2021 \cite{soda21}. 
In this variant, the cost of an online algorithm is compared to the cost of a static offline algorithm that is not allowed to perform any communication. Rather, requests must induce a communication graph which can be perfectly partitioned among the servers:
if we define a demand graph $G_\sigma$ over the set of processes and edges $(p_i,p_j)$ for every request between $p_i$ and $p_j$
within $\sigma$, then the connected components of $G_\sigma$ can be mapped to servers so that there is no inter-server edge. 

For this setting, it is known that no online algorithm achieving a low competitive ratio can exist without resource augmentation:
there is a lower bound of $\Omega(n)$ for any deterministic online algorithm~\cite{disc16,sidma19}.
In fact, currently the best known upper bound for this problem variant is nearly quadratic, namely  $O(n^{23/12})$ ~\cite{stacs24}.

%On the positive side, deterministic online algorithms have been presented which achieve a competitive ratio of FIXME with resource augmentation FIXME. 
%It has also been shown that even with augmentation of FIXME, no deterministic online algorithm can achieve a low competitive ratio better than FIXME.

\subsection{Our Contributions and Novelty}

We revisit the online graph partitioning problem in the learning model. Compared to prior work, we consider a more powerful adversary which is not restricted to keep processes on the same server once they communicated; rather, algorithms are allowed to communicate across servers.
This model is not only more practical but also motivated as a next step toward solving the general online balanced partitioning problem. 
Since the lower bound of $\Omega(n)$ for algorithms without resource augmentation of course also holds in this more general adversarial model, we in this paper study models with augmentation, as usual in the related work. 

Our main contribution are intriguing tight bounds for this problem. 
We first show that the more general model is significantly more challenging as the lower bounds increase substantially. In particular, we show that even under arbitrary resource augmentation and even when compared to a static offline algorithm, any deterministic online algorithm has a competitive ratio of at least  $\Omega(\sqrt{k})$. 

\begin{restatable}{theorem}{LowerBoundOneEps}\label{theorem:lower_bound_unlimited_aug}
For any deterministic online algorithm $\ALG$ with unlimited resource
augmentation, there exists a request sequence $\sigma$, such that
$\ALG(\sigma)\ge\Omega(\sqrt{k})\cdot\OPTS(\sigma)$, where $\OPTS$ is the optimum static offline algorithm. 
\end{restatable}

The proof is based on a key observation regarding a scenario with two sets
of processes $A$ and $B$ that initially are located on distinct servers. 
Given any deterministic algorithm $\ALG$ we construct a request sequence such
that $\ALG$
performs poorly on requests between processes from $A\cup B$. 

We then consider scenarios with only small resource augmentation and show an even higher lower bound, which again even holds when compared to a static offline algorithm:

\begin{restatable}{theorem}{LowerBoundsOneEps}\label{theorem:lower_bounds_one_eps}
For any $k \ge 10$, any constant $\epsilon \le \frac{1}{10}$, and any deterministic algorithm $\ALG$ with augmentation $1+\epsilon$, there is a request sequence $\sigma$, such that $\ALG(\sigma) 
\ge\Omega(\max(\sqrt{k\ell \log k}, \ell \log k))\cdot\OPTS(\sigma)$, where $\OPTS$ is the optimum static offline algorithm. 
\end{restatable}

Also this lower bound shows that the more general learning model is challenging. 
In \cite{soda21}, it was shown that in the standard learning model
deterministic algorithms with  resource augmentation $(1+\epsilon)$ have a
competitive ratio of at least $\Omega(\ell \log k)$, but also an upper bound of
$\mcO(\ell \log k)$ has been derived.

The proof of this lower bound is fairly technical and builds upon ideas of our simpler lower bound above
as well as techniques from the lower bound in \cite{soda21}.

%untere schranken: für 1+eps besonders interessant, aber auch 2+eps. habe connected compos in demand-graphs: wenn nicht monochromatisch dann hohe Kosten, das ist insight lower bound von SODA. funktioniert hier nicht, hier brauche von LP.

We then present almost optimal online algorithms which match our lower bounds.
In our approach, we formulate the generalized learning problem as a linear program. 
We maintain a dual solution to the LP and bound the cost of our algorithm in terms of the dual profit.
For this we reformulate the generalized learning problem as a graph problem: for a given request sequence 
we define a time-expanded graph. 
A key novelty of our approach is that this time-expanded graph is used in an online setting, where dual solutions
are computed on the fly. 
Our online algorithms then consist of two parts, a clustering part and a scheduling part. 
The clustering part maintains a partition of the process set into disjoint pieces (clusters). 
The scheduling part of an algorithm then decides which
cluster is mapped to which server, however, without splitting clusters across different servers. 

Based on these algorithmic insights, we first show that very little augmentation (namely $1+\epsilon$) is sufficient to achieve a competitive ratio of $\mcO(\max(\sqrt{k\ell \log k}, \ell \log k))$. 
In particular, we prove the following theorem:
\begin{restatable}{theorem}{UpperBoundsOneEps}
    There exists a deterministic online algorithm for the generalized learning
    problem with competitive ratio $\mcO(\max(\sqrt{k\ell \log k}, \ell \log
    k))$ and augmentation $1+\epsilon$.
\end{restatable}

This algorithm uses our flow clustering procedure and uses the algorithm from the standard learning model
presented in \cite{soda21} as a subroutine in a second phase. In the second phase, hence entire connected components 
are assigned to a single server. However, we will show that the cost incurred during this phase is not
significantly higher than that of the first phase.

We then present a second online algorithm which uses slightly more augmentation of $2+\epsilon$ to achieve a significantly better
competitive ratio of $\mathcal{O}(\sqrt{k})$, hence removing the linear factor $\ell$,
which also matches our corresponding lower bound. Also this algorithm relies on our
flow clustering procedure, however, since our goal is to achieve a $\mcO(\sqrt{k})$ competitive ratio, we cannot simply apply the procedure as before. 
Concretely, we derive the following theorem:
\begin{restatable}{theorem}{UpperBoundsTwoEps}
    There exists a deterministic online algorithm for the general learning problem with competitive ratio $\mcO(\sqrt{k})$ and augmentation $2+\epsilon$.
\end{restatable}

%obere schranke: wir machen clustering mit flow routine (also wie dynamisch zerstückelt) und machen scheduling, wie im SPAA. flow proceduce ist new und gut wenn OPT klein, hat Kosten immer max OPT^2. drum brauche zusätzliche Ideen: wenn das nicht gilt, nehmen wir 1+eps algo von SODA (combine with SODA algo) und für 2+eps: stoppe procedure und halte connected compo zusammen (dann trivialen greedy Algo). time-expanded graphs zum ersten mal für online (hier brechnen wir on the fly eine duale lösung). 

%discuss in related work: in Sigmetrics/SODA darf online nicht splitten macht
%aber keinen Unterschied für deterministischen aber für randomisierten könnte es Unterschied machen. bei unterer Schranke von competitive ratio ev. nicht mehr von SODA. 

%It finds
%applications in the context of resource allocation in the cloud
%and for optimizing distributed data structures such as union–
%find data structures.

\subsection{Further Related Work}
\label{onl:sec:related}

The dynamic balanced (re-)partitioning problem was introduced by Avin et
al.~\cite{disc16,sidma19}. For a general online setting where requests can be
arbitrary over time, the paper showed a linear in $\Omega(k)$ lower bound (even if the online algorithm
is allowed significant resource augmentation), and also presented 
a deterministic online
algorithm which achieves a competitive ratio of $O(k \log k)$ (with constant augmentation). That algorithm however relies on expensive repartitioning
operations and has a super-polynomial runtime. Forner et al.~\cite{apocs21repartition} later showed that a  competitive ratio of $O(k \log k)$
can also be achieved with a polynomial-time online algorithm which monitors the
connectivity of communication requests over time, rather than the density.
Pacut et al.~\cite{infocom21repartitioning} presented an
$O(\ell)$-competitive online algorithm for a scenario without resource
augmentation and the case where $k = 3$.
The dynamic graph partitioning problem has also been studied in scenarios with less resource augmentation or
where augmentation is even strictly forbidden.
Regarding the former, Rajaraman and Wasim~\cite{rajaraman2022improved} presented a
$O(k\ell \log k)$-competitive algorithm for scenarios with only slight resource augmentation.
Regarding scenarios without augmentation, Avin et al. \cite{disc16} already presented a simple algorithm achieving a competitive ratio of 
$O(k^2 \cdot \ell^2)$, i.e., $O(n^2)$. Only recently a first online algorithm achieving a subquadratic competitive ratio of $O(n^{23/12})$ (ignoring polylog factors) has been presented~\cite{stacs24}.
An improved analysis of the original algorithm of Avin et al.~was presented by Bienkowski et al.~\cite{bienkowski2021improved}, by translating the problem to a system of linear integer equations and using Graver bases.
The problem has further also been studied from an offline perspective by Räcke et al.~who presented a polynomial-time $O(\log n)$-approximation algorithm \cite{racke2022approximate}, using LP relaxation and Bartal’s clustering algorithm to round~it.

Deterministic online algorithms have also been studied for specific communication patterns, namely the ring. In~\cite{obr-ring,netys17learn}, the adversary generates the communication sequence from an arbitrary (adversarially chosen) random distribution in
an \emph{i.i.d.} manner~\cite{obr-ring,netys17learn} from the ring. In this scenario, it has been shown that even deterministic algorithms can achieve a polylogarithmic competitive ratio.  In~\cite{racke2023polylog}, a more general setting is considered where the adversary can choose edges from the ring arbitrarily, in a worst-case manner. It is shown that a polylogarithmic competitive ratio of $O(\log^3 n)$ can be achieved in this case by randomized algorithms; against a static solution, the ratio can be improved to  $O(\log^2 n)$ (this ratio is strict, i.e., without any additional additive terms).
The problem on the ring is however very different from ours and the corresponding algorithms and techniques are not applicable in our setting. 

The learning variant considered in this paper was introduced by Henzinger et al.~\cite{sigmetrics19learn,soda21}. The authors presented a deterministic exponential-time algorithm with competitive ratio $O(\ell \log \ell \log k)$ as well as a lower bound of
$\Omega(\log k)$ on the competitive ratio of any deterministic online algorithm.
While their derived bounds are tight for $\ell=O(1)$ servers, 
there remains a gap of factor $O(\ell \log \ell)$ between upper and lower bound
for the scenario of $\ell=\omega(1)$.
In~\cite{soda21}, Henzinger et al.~present deterministic and randomized algorithms which achieve (almost)
tight bounds for the learning variant. In particular, a
polynomial-time randomized algorithm is described which achieves a polylogarithmic competitive
ratio of $O(\log \ell + \log k)$; it is proved that no randomized online
algorithm can achieve a lower competitive ratio. Their approach establishes and
exploits a connection to generalized online scheduling, in particular, the
works by Hochbaum and Shmoys~\cite{hochbaum87using} and Sanders et
al.~\cite{sanders09online}. 
In our paper, we consider a more general variant of the learning model, where the algorithms are not restricted to keep process pairs, once they communicated, on the same server forever. 

More generally, dynamic balanced partitioning is related to 
dynamic bin packing problems which allow for limited \emph{repacking}~\cite{FeldkordFGGKRW18}:
this model can be seen as a variant of our problem where pieces (resp.~items)
can both be dynamically inserted and deleted, and it is also possible to open new
servers (i.e., bins). The goal is to use only an (almost) minimal number of servers, and
to minimize the number of piece (resp.~item) moves.
However, the techniques of~\cite{FeldkordFGGKRW18} do not extend to our problem.
Another closely related (but technically different) problem is \emph{online vertex recoloring (disengagement)}~\cite{azar2022competitive,rajaraman2024competitive}. 
This problem is the flip side of our model: in disengagement, the processes in requests should be 
are \emph{separated} and put on different servers (rather than collocated). 

\subsection{Future Work}

We see our work as a stepping stone toward solving the general online balanced partitioning problem. While we remove the restriction that the algorithm needs to keep communication pairs on the same server, we still require that the graph induced by the communication requests can be perfectly partitioned across the servers (i.e., without inter-server edges). In other words, while our paper removes the restriction on the adversary, it remains to remove this restriction on the input as well. In particular, it remains an open questions whether a polylogarithmic competitive ratio can be achieved in general setting by randomized algorithms.

%\subsection{Organization}
%
%The remainder of this paper is organized as follows.
% We introduce our model and given an overview of our results in Section~\ref{sec:model}.
% The dynamic model is studied in Section~\ref{sec:dyn} and the static model in  Section~\ref{sec:static}. We conclude our paper in  Section~\ref{sec:conclusion}. 
%FIXME

\section{Model and Preliminaries}

We now present our formal model in details and also introduce the necessary preliminaries.

Let $\ell$ denote the number of servers and $k$ the \emph{capacity} of a
server, i.e., the maximum number of processes that can be scheduled on a single
machine. We use $V=\{p_1, p_2, \dots, p_{n}\}$ with $n\le\ell k$ to denote the
set of processes. In each time step $t$ we receive a request
$\sigma_t=\{p_i, p_{j}\}$ with $p_i,p_{j} \in V$, which means that these two
processes communicate.
For a request sequence $\sigma$ we define the \emph{demand graph} $G_\sigma$
as a (multi-)graph that contains a vertex for every
process and an edge $(p_i,p_j)$ for every request between $p_i$ and $p_j$
within $\sigma$.

Serving a communication request costs exactly $1$ if both requested processes
are located on different servers, otherwise~0. We call this the
\emph{communication cost} of the request.
Before the communication, an online algorithm may (additionally) decide to
perform an arbitrary number of migrations. Each migration of a process to
another server induces a cost of $1$ and contributes to the
\emph{migration cost} of the request. 

An assignment of processes to servers is referred to as a (current)
\emph{configuration} or \emph{scheduling} of the algorithm. We start in the
initial configuration $\mathcal{I}$ where each server contains exactly $k$ processes. The
initial server of a process $p$ is called the \emph{home server} of $p$ and
denoted with $h(p)$. 
%Additionally, we associate each server with a unique color.
%\ruslan{server colors?}

In the end, after performing all migrations, each server should obey its capacity
constraint, i.e., it should have at most $k$ processes assigned to it. The goal
is to find an online scheduling of processes to servers for each time step that
minimizes the sum of migration and communication costs and obeys the capacity
constraints.

\paragraph*{Resource augmentation.}
We usually consider online algorithms that do not have to fulfill the capacity
constraints exactly but are allowed to have a little bit of slack. Whereas
the optimum offline algorithm may schedule at most $k$ processes on any server
at any given time, the online algorithm may schedule $\alpha k$ processes on any
server for some parameter $\alpha > 1$. We say the online algorithm
uses resource augmentation $\alpha$.

\paragraph*{Restricting the request sequence.}
%\harry{maybe part of this has to got to the intro}
In the \emph{learning variant} of the problem\cite{sigmetrics19learn, soda21}
the request sequence is restricted so as to allow for a static solution that does not
require any communication cost. This means there exists a very good solution 
(in terms of communication cost) and the goal is to \emph{learn} this solution 
as the requests appear. 

Henzinger et al.~\cite{sigmetrics19learn} introduced this learning variant and
compared the cost of an online algorithm to the cost of a \emph{static} offline
algorithm that does not perform any communication. This means initially the
optimum algorithm (which knows the request sequence in advance) chooses a
static placement that does not require any communication. It then switches to
this placement paying some migration cost, and then serves the request 
sequence without any further migration and/or communication.
%\harry{this is not really exact; for simplicity I simply modeled the standard
%  learning setting has having a static offline algorithm that has no communication}

In this paper we introduce the \emph{generalized learning model} where we have
the same restriction on the request sequence but we lift some of the
restrictions on the optimal offline algorithm that we compare against. We
consider two scenarios. In the first scenario the offline algorithm is
\emph{static}, i.e., it must initially migrate to a placement and then serve
the request sequence without changing the placement of the processes. However,
in contrast to the standard learning model~\cite{soda21} this placement does not need
to be perfect in the sense that serving the request sequence does not incur
communication.

In the second scenario the offline algorithms is completely dynamic and only
has to obey capacity constraints on the servers.

\paragraph*{Comparing static and dynamic adversaries.}
Let $\OPTS$ and $\OPTD$ denote the optimum \emph{static} and
\emph{dynamic} offline algorithm, respectively. For a request sequence $\sigma$
we use $\OPTS(\sigma)$ and $\OPTD(\sigma)$ to denote the cost of these
algorithms on request sequence $\sigma$. We drop $\sigma$ if it is clear from the
context and, hence, use $\OPTS$ and $\OPTD$ to denote an algorithm and also its cost.
The following simple fact holds because the dynamic algorithm is more powerful
than its static counter-part.
\begin{fact}For any request sequence $\sigma$ we have
$\OPTS(\sigma) \ge \OPTD(\sigma)$.
\label{fact:static_dynamic}
\end{fact}

Next, we show that there exist request sequences for which the optimal static
algorithm performs significantly worse than the optimal dynamic algorithm.
\begin{lemma}
For any $k$ there exists a request sequence $\sigma$% ($|\sigma|=\Theta(k)$),
where $\OPTS \ge \Omega(k)\cdot\OPTD$. Moreover, this inequality holds even if
$\OPTS$ has unlimited resource augmentation whereas $\OPTD$ has no resource
augmentation at all.
\end{lemma}
\begin{proof}
We fix two servers $s_p$ and $s_q$ with $p_1, \dots, p_k$ initially residing on
$s_p$, and $q_{1}, \dots, q_{k}$ initially residing on $s_q$. We generate the
following request sequence of length $2\lfloor \frac{k}{2} \rfloor$:
\[\sigma = \{p_1, q_{1}\}, \{p_1, q_{2}\}, \dots, \{p_1, q_{\lfloor \frac{k}{2} \rfloor}\}, \{p_1, p_{2}\}, \{p_1, p_{3}\}, \dots, \{p_1, p_{\lfloor \frac{k}{2} \rfloor + 1}\}\] 

Let $C$ be the connected component in the demand graph $G_\sigma$. $C$ contains
$\lfloor \frac{k}{2} \rfloor$ processes of server $s_p$ and
$\lfloor \frac{k}{2} \rfloor$ processes of server $s_q$. Every optimal static
algorithm (even with unlimited resource augmentation) will either move $\Omega(k)$
processes or pay $\Omega(k)$ communication cost. However, a dynamic algorithm
that first swaps $p_1$ and $q_k$, serves the first $k-1$ requests, swaps $p_1$
and $q_k$ again, and serves the remaining requests, pays only a cost of 2
for the swaps.
\end{proof}
The above lemma means that there may be a very large gap between the static and
dynamic model. Nevertheless, we will show matching upper
and lower bounds, where the lower bounds are against a static adversary and the
upper bounds are against a dynamic adversary.

\section{Lower Bounds}

In \cite{soda21}, it was shown that in the standard learning model
deterministic algorithms with  resource augmentation $(1+\epsilon)$ have a
competitive ratio of at least $\Omega(\ell \log k)$, and deterministic
algorithms with  resource augmentation $(2+\epsilon)$ have a competitive ratio
of at least $\Omega(\log k)$. 
%Moreover, randomized algorithms with
%$(1+\epsilon)$ resource augmentation have a competitive ratio of at least
%$\Omega(\log k + \log \ell)$. \ruslan{do we need to mention rand algo?} 
Note
that the $\mcO$-notation is hiding dependencies on the parameter $\epsilon$.

These lower bounds directly carry over to the generalized model since we only
make the adversary more powerful. We now show that this change in the model 
makes an important difference because the lower bounds increase substantially.

%Since we can simulate the standard learning model simply by repeating the
%requests (thus forcing the optimal algorithm to schedule the connected
%components on the same server), the lower bounds directly translate to our
%generalized model. However, we will now show that the general learning
%model is even more challenging than the standard learning model.

%\ruslan{the notation in the next two sections is slightly different, need to fix that}
\subsection{Lower Bounds for Algorithms with Unlimited Augmentation}
In this section we show a lower bound of  $\Omega(\sqrt{k})$ even for the case
with unlimited augmentation. This is in stark contrast to the 
the standard learning model in the literature (which we generalize in this paper) where already
an augmentation of $2+\epsilon$ allowed an upper bound of
$\mcO(\log k)$.

\LowerBoundOneEps*
\begin{proof}
Fix a deterministic algorithm $\ALG$ and two arbitrary servers. Assume Server~1
contains processes $S_1 = \{p_1, p_2, \dots, p_{k/2}\}$, and Server~2 contains
processes $S_2 = \{p_{k/2+1}, p_{k+2}, \dots, p_{k}\}$ (wlog.\ assume that $k$
is even). As long as not all
processes $p_1,\dots,p_{k}$ are on the same server we issue a request of the
form $\{p_i,p_{i+1}\}$, $i\in\{1,\dots,k-1\}$ for which $p_i$ and $p_{i+1}$ are
on different servers. $\ALG$ incurs cost for each of these  requests. Hence, its cost is
at least $\max\{|\sigma|,k/2\}$ because it has to move $k/2$ processes in order to
stop the sequence.
%\harry{this uses that we first have to pay for communication and then we can move. Is this correct in our model?}

The optimum algorithm $\OPT$ identifies a pair $(p_i,p_{i+1})$ in the 
range $i\in\{k/2-s,k/2+s\}$ ($s<k/2$) that experiences the least number of requests. 
The range contains $2s+1$ pairs, and, hence, one of them has at most
$|\sigma|/(2s+1)\le |\sigma|/(2s)$ requests. Regardless, of which pair
attains the minimum, $\OPT$ can move at most $2s$ processes between Server~1
and Server~2 so that both servers have $k$ processes and only requests 
for this minimum pair requires communication (this uses the fact that servers
have enough \emph{dummy} processes that are not involved in any requests and
can be moved between servers so as to balance the load). Consequently, $\OPT$
can serve the sequence with cost $2s+|\sigma|/(2s)$. Choosing
$s=\Theta(\sqrt{k})$ gives the theorem.
\end{proof}

\mathversion{bold}
\subsection{Lower Bounds for Algorithms with Augmentation $1+\epsilon$}
\mathversion{normal}
For very small resource augmentation, 
in the standard learning model, there is a lower bound of
$\Omega(\ell\log k)$ for any online algorithm, i.e., there is a strong dependency on
the number of servers. In this section, we show that in the generalized learning model,
a significantly higher lower bound of $\Omega(\max(\sqrt{k\ell \log k}, \ell \log k))$ can be shown. 

The proof follows a structure similar to that presented in \cite{soda21}. However, we enhance it by integrating key insights from both their approach and Theorem~\ref{theorem:lower_bound_unlimited_aug}. This combination of ideas allows us to establish stronger lower bounds in the generalized learning model. See %Appendix~\ref{appendix_lower_bounds_one_eps} for the proof of the following theorem.

\LowerBoundsOneEps*

%\mathversion{normal}
In the following, we assume that $\epsilon\le 1/10$ and 
$\epsilon k=2^m$ for a positive integer $m$.

\paragraph{Notation}
%Let $m$ be a positive integer such that $\epsilon k = 2^m$. % As $\epsilon \le 1/32$ it follows that $k \ge 2^{m+5}$. 
Fix any deterministic algorithm $\ALG$ with augmentation $1+\epsilon$. We
construct a request sequence $\sigma$ and analyze the costs $\ALG(\sigma)$ and
$\OPTS(\sigma)$. 
In the following we refer to requests as edges.

We associate each server
with a color and each process with the color of its home server. We say,
the \emph{main server} of a color $c$ is the server that, out of all servers,
currently contains the largest number of processes with color $c$ (ties broken
arbitrarily).

Each server reserves $\epsilon k$ of its initial processes as \emph{padding
  processes}. These processes will not be included in any request, so migrating
them will not result in any communication cost for $\OPTS$. These processes are
only necessary so that $\OPTS$ can balance the load
between servers and create a static offline solution without
resource augmentation.

We construct the request sequence $\sigma$ such that the connected components
of the demand graph $G_\sigma$ mostly consist of processes of the same color. In
fact, there will only be one component, the \emph{special component}, 
that contains different colors. 
The other non-padding processes will be part of \emph{colored components}. At 
each point in time a color $c$ has a rank $r_c$. All but (at most) one of its
components will have the same size $2^{r_c}$. We call a $c$-colored components
that has size $2^{r_c}$ a \emph{regular component of color $c$}. A $c$-colored component 
that does not have size $2^{r_c}$ is called \emph{the extra component for color $c$}.
We call processes in extra components \emph{extra processes} and  processes in regular
components \emph{regular processes}.

During the request sequence, we merge components by issuing a request between
processes from different components. To start, we create the special component.
Let $S_1 = \{p_1,\dots,p_{3\epsilon k}\}$ be arbitrary $3\epsilon k$ processes
initially stored on Server~1 and
$S_2=\{p_{3\epsilon k+1},\dots,p_{6\epsilon k}\}$ arbitrary $3\epsilon k$
processes initially stored on Server~2. We say processes $p_i$ and $p_{i+1}$
are consecutive processes. We merge $S_1$ and $S_2$ into the special component
by issuing requests $(p_i,p_{i+1})$ for $i\in 1,\dots,6\epsilon k-1$. 
All other non-padding processes initially form regular components of size 1. 

We will ensure that $\ALG$ stores each component on a single server by
repeatedly requesting processes of that component that are stored on different
servers. On the other hand, $\OPTS$ is not restricted to storing each component
on a single server. It will first perform some migrations and then stay in this
configuration for the entire request sequence, possibly incurring some
communication cost.

Most components of a color $c$ will be scheduled on the main server for color
$c$. We call a color $c$ \emph{spread-out} if the number of processes in
$c$-colored components that are not located on the main server for $c$ is at
least $\epsilon k$ (note that regardless of their color, special and padding processes are not considered to
be in a $c$-colored component).

\paragraph*{The request sequence.}
We generate the request sequence in iterations. 
In the first phase of each iteration, as long as there is a component $C$ that
is split across servers, we request a process pair from $C$ that resides on
different servers. If $C$ is a regular component or an extra component, 
we choose this pair arbitrarily.
However, if $C$ is the special component, we choose an arbitrary pair
$(p_i,p_{i+1})$ that is located on different servers.

The first phase of an iteration ends when $\ALG$ schedules all processes of any
component entirely on the same server.

In the second phase, we proceed as follows: According to
Lemma~\ref{lemma:lower_bounds_def}, there must exist a \emph{spread-out} color
$c$. Let $r_c$ be its rank. We find a matching $M_c$ between the regular
components of color $c$ that maximizes the number of edges that cross between
servers. Then we merge the matched components. By construction, the size of the
newly created components is $2^{r_c+1}$. If the number of regular components
is odd, there will be an unmatched component. We merge this component with the
extra component. Before the
merge, the size of the leftover component was less than $2^{r_c}$; thus, after
the merge, its size will be less than $2^{r_c+1}$. The rank of
color $c$ increases by one. 

This finishes the iteration for constructing the request sequence. 
The construction finishes once the rank of each color is at least $m$.

\subsubsection*{Cost Analysis}
%Now, we show that there always exist a deficient color. 
\begin{lemma}
At the start of the second phase of an iteration, there exists a spread-out color.
\label{lemma:lower_bounds_def}
\end{lemma}
\begin{proof}
Recall that a color $c$ is spread-out if the number of processes in $c$-colored 
components that are not located on the main server for $c$ is at least $\epsilon k$.
A color can have at most $3\epsilon k$ special processes and has exactly
$\epsilon k$ padding processes. This means that it has at least $k-4\epsilon k$
processes in $c$-colored components. 

Assume for contradiction that all colors have strictly more than $(1-5\epsilon) k$ non-special
non-padding processes at their main server (these are the processes in
\emph{colored components}). Since $\epsilon \le {1}/{10}$, we
have $(1-5\epsilon) k \ge k/2$. Thus, a server can be the main server for only
one color, which means that there is a one-to-one mapping of colors to main
servers.

Now, consider the color $c$, whose main server currently holds the special
component. The size of the special component is $6\epsilon k$. Hence, there are
at most $(1+\epsilon)k - 6\epsilon k = (1-5\epsilon)k$ non-special
non-padding processes at this server. This gives a contradiction.
\end{proof}

The next two lemmas derive bounds on the costs of $\ALG$ and $\OPTS$. Let $R$
denote the number of requests generated due to the special component being
split across servers.
\begin{lemma}
The cost of $\OPTS$ on the request sequence $\sigma$ is at most $12\min(\sqrt{R}, \epsilon k)$.
\label{lemma:lower_bounds_opt}
\end{lemma}
\begin{proof}
We analyze two cases: 
\begin{enumerate}
\item $\sqrt{R} \le \epsilon k$:\\
$\OPTS$ can proceed as follows. Within the special component it identifies an
edge that is requested the fewest number of times and lies within
the sub-sequence $p_{3\epsilon k-\sqrt{R}},\dots,p_{3\epsilon k+\sqrt{R}}$ of 
processes. As the total number of
(special) requests is $R$ and the range contains at least $\sqrt{R}$ edges there must
exist an edge that is requested at most $\sqrt{R}$ times. Let this be the
edge $(p_i,p_{i+1})$. $\OPTS$ can now shift at most $\sqrt{R}$ processes between
server $S_1$ and $S_2$ so that processes $p_1,\dots,p_i$ are on $S_1$ and
processes $p_{i+1},\dots,p_{6\epsilon k}$ are on server $S_2$. It then uses the
padding processes to rebalance the load of $S_1$ and $S_2$ (another movement of
at most $\sqrt{R}$ processes). Hence, it incurs a movement cost of $2\sqrt{R}$
and a communication cost of $\sqrt{R}$ for the requests between $p_i$ and $p_{i+1}$.

The remaining requests (non-special requests) are all between processes that are
on the same server. Hence, $\OPTS$ does not pay anything for these requests. 
This gives $\OPTS(\sigma)=3\sqrt{R}$.

\item $\sqrt{R} > \epsilon k$:\\
Let $c^*$ be the last color that reaches the rank of $m$ in the sequence
$\sigma$. $\OPTS$ initially moves the special
component to the $c^*$-colored server $s$.
The number of non-special non-padding processes
on $s$ is at least $k - 3\epsilon k - \epsilon k \ge 6\epsilon k$, where the 
inequality holds because $\epsilon \le 1/10$. These are all processes in 
$c^*$-colored components.

 Since the rank of $c$ is $m$, the size of all $c^*$-colored
components except the extra component is exactly $2^m=\epsilon k$. Since the
size of the extra component is less than $\epsilon k$, there must exist at
least six $c^*$-colored components of size exactly $\epsilon k$ on $s$. We move three
of these components to Server 1 and the other three to Server 2, so that
every server contains exactly $k$ processes. Notice that all requests are
created between processes of the same component and no component is split
between different servers. Thus, $\OPT$ pays only for the
migration cost, and therefore its cost is at most $12\epsilon k$.
\end{enumerate}
In summary, $\OPTS$ pays at most $12\min(\epsilon k, \sqrt{R})$.
\end{proof}

\begin{lemma}
The online algorithm $\ALG$ has cost $\ALG(\sigma)\ge R + \epsilon k \ell \log(\epsilon k)$.
\label{lemma:lower_bounds_alg}
\end{lemma}
\begin{proof}
\newcommand{\CS}{\ensuremath{\operatorname{cost}_{\operatorname{s}}}}
\newcommand{\CR}{\ensuremath{\operatorname{cost}_{\operatorname{r}}}} 
We can divide the cost of $\ALG(\sigma) = \CS + \CR$ into two parts: $\CS$ represents
the cost due to the requests within the special component, and $\CR$ represents
the cost due to the requests within regular components. Each request results in
a cost of $1$ for $\ALG$, either due to migration or communication. Therefore,
$\CS \ge R$.

Now, we derive a lower bound on $\CR$. Since we issue requests until $\ALG$
puts all processes of a component on a single server, $\ALG$ has to pay at least
the size of the smaller component each time two components on different servers
merge. Hence, if two such components of size $2^r$ merge, $\ALG$ pays at least
$2^r$.

According to Lemma \ref{lemma:lower_bounds_def} in each iteration we find a
spread-out color $c$. Thus, there are at least $\epsilon k = 2^m$ processes of
color $c$ that are not on their main server. Let $r$ be the rank of $c$. Since
the size of the extra component is less than $2^r$, there must exist
$2^{m-r}$ regular $c$-colored components of size exactly $2^r$ that are not on
their main server. We fix these $2^{m-r}$ components. Now, we greedily pick for
each such component a unique matching partner that is not on the same server.
Since for every component there are at least $2^{m-r}$ candidate partners, we
can easily find such a matching. Thus, the size of the
matching $\mathcal{M}_c^r$ must be at least $2^{m-r}$, and since each merge
results in a cost of $2^r$ for $\ALG$, its overall cost in every iteration is
at least $2^{m-r} 2^r = \epsilon k$.

There are $\ell$ colors, and during each iteration the rank of one color
increases by one. The algorithm stops when the rank of each color is at least
$m$. Therefore, there are at least $\ell m = \ell\log(\epsilon k)$ iterations.
Hence, $\CR \ge \epsilon k \ell \log(\epsilon k)$. Summarizing,
$\ALG \ge R + \epsilon k \ell \log \epsilon k$.
\end{proof}

Finally, we can state the theorem:
\LowerBoundsOneEps*
\begin{proof}
Let $C$ be the competitive ratio of $\ALG$. Due to Lemmas
\ref{lemma:lower_bounds_opt} and \ref{lemma:lower_bounds_alg} 
we know that 
$C \ge (R + \epsilon k \ell \log (\epsilon k))/(12\min(\epsilon k,
  \sqrt{R})$. We show in separate steps that
$C \ge \frac{1}{12}\ell \log (\epsilon k)$ and
$C \ge \frac{1}{12}\sqrt{\epsilon k\ell \log (\epsilon k)}$.

\begin{enumerate}[a)]
  \item Clearly, $C \ge (R + \epsilon k \ell \log \epsilon k)/(12\epsilon k) \ge \frac{1}{12} \ell \log \epsilon k$.
  \item Next, we analyze two subcases:
    \begin{enumerate}[i)]
        \item $R > \epsilon k \ell\log \epsilon k$. Then: 
        \[ C \ge \frac{R + \epsilon k \ell \log \epsilon k}{12\sqrt{R}} \ge \frac{1}{12}\sqrt{R} > \frac{1}{12}\sqrt{\epsilon k \ell\log \epsilon k} \]
        \item $R \le \epsilon k \ell\log \epsilon k$. Then: 
        \[ C \ge \frac{R + \epsilon k \ell \log \epsilon k}{12\sqrt{R}} \ge  \frac{\epsilon k \ell \log \epsilon k}{12\sqrt{R}} \ge \frac{1}{12} \sqrt{\epsilon k \ell\log \epsilon k}\]
    \end{enumerate}
    The above two subcases imply that $C \ge \frac{1}{12} \sqrt{\epsilon k\ell \log \epsilon k}$.
\end{enumerate}
The two cases directly imply the lemma.
\end{proof}
\section{Algorithmic Concepts}\label{sec:flow}
In this section, we develop basic algorithmic tools that will help us solve the
generalized learning problem. Our approach involves formulating the generalized
learning problem as a linear program. We maintain a dual solution to the LP
and bound the cost of our algorithm in terms of the dual profit. For this we
reformulate the generalized learning problem as a graph problem.
%In the following, we refer to processes as xertices.

\def\TG{\bar{G}}
\subsection{Time-Expanded Graph}
For a given request sequence $\sigma$ we define the \emph{time-expanded graph}
$\TG_\sigma$ as the follows. We introduce one node for every
server $s\in S$, and $|\sigma|$ nodes $p^1,\dots,p^{|\sigma|}$ for every
process $p\in V$. The first set of nodes are called \emph{server nodes} and the
nodes $p^i$, $i\ge 1$ are called \emph{process time nodes}. For
simplicity we use $p^0$ for a process $p$ to denote the server node that in the
initial configuration holds process $p$ (thus, a server node $s$ in
$\TG_\sigma$ may have many different names; at least one for every process
initially located on $s$).

The edges of the undirected graph $\TG_\sigma$ are defined as follows.
For every process $p$ and time-step
$t\in\{1,\dots,|\sigma|\}$ we add an edge $(p^{i-1}, p^{i})$. In
addition, we add an edge $(p^t,q^t)$ if the $t$-th request
is between processes $p$ and $q$. We call the first set of edges
\emph{migration edges} and the second set \emph{communication edges}.
Figure \ref{fig:time_expanding_graph} illustrates the
time-expanded graph.

\tikzstyle{server}=[fill=black!10]
\tikzstyle{annotation}=[anchor=north,fill=none,shape=rectangle,inner
sep=2pt,line width=0pt,draw=none,minimum size=0pt, font=\footnotesize]
\begin{figure}[h]
    \centering
\begin{tikzpicture}[every node/.style={draw, fill=blue!10,line width=1pt, circle, inner sep = 0mm, minimum size=.8cm, font=\small}, scale=0.65]

  \def \dx {1.8}
  
  % First row of nodes (p_x^1)
  \foreach \x in {1,2,...,12} {
    \node (p1\x) at (\x*\dx,0) {$p_{\x}^1$};
  }

  % Second row of nodes (p_x^2)
  \foreach \x in {1,2,...,12} {
    \node (p2\x) at (\x*\dx,2) {$p_{\x}^2$};
  }

    % Third row of nodes (p_x^2)
  \foreach \x in {1,2,...,12} {
    \node (p3\x) at (\x*\dx,4) {$p_{\x}^3$};
  }

  % Fourth row of nodes (p_x^2)
  \foreach \x in {1,2,...,12} {
    \node (p4\x) at (\x*\dx,6) {$p_{\x}^4$};
  }

  % nodes (s_x)
  \coordinate (ds) at (0, -2.5);
  
  \node[server] (s1) at ($0.5*(p11)+0.5*(p14)+(ds)$) {$s_1$};
  \foreach \x in {1,2,...,4} {
    \draw[line width=1pt] (s1) -- (p1\x);
  }

  \node[server] (s2) at ($0.5*(p15)+0.5*(p18)+(ds)$) {$s_2$};
  \foreach \x in {5,6,...,8} {
    \draw[line width=1pt] (s2) -- (p1\x);
  }

  \node[server] (s3) at ($0.5*(p19)+0.5*(p112)+(ds)$) {$s_3$};
  \foreach \x in {9,10,...,12} {
    \draw[line width=1pt] (s3) -- (p1\x);
  }

h  % Connecting nodes in consecutive rows
  \foreach \x in {1,2,...,12} {
    \draw[line width=1pt] (p1\x) -- (p2\x);
    \draw[line width=1pt] (p2\x) -- (p3\x);
    \draw[line width=1pt] (p3\x) -- (p4\x);
  }

  % request edges 
  \draw[blue,line width=2pt] (p12.north east) to[bend left=45] (p13.north west);
  \draw[blue,line width=2pt] (p24.north east) to[bend left=45] (p26.north west);
  \draw[blue,line width=2pt] (p36.north east) to[bend left=25] (p39.north west);
  \draw[blue,line width=2pt] (p41.north east) to[bend left=20] (p49.north west);

  \node[annotation] at (s1.south) {$p_1^0=p_2^0=p_3^0=p_4^0$};
  \node[annotation] at (s2.south) {$p_5^0=p_6^0=p_7^0=p_8^0$};
  \node[annotation] at (s3.south) {$p_9^0=p_{10}^0=p_{11}^0=p_{12}^0$};
\end{tikzpicture}

\caption{An illustration of the time expanded graph $\TG_\sigma$ for $\ell=3$ servers, 
server capacity $k=4$, and $|\sigma|=4$ requests. The bottom row are the server
nodes, the remaining nodes are process time nodes.
The request sequence is $\sigma = (\{p_2, p_3\}, \{p_4, p_6\}, \{p_6, p_9\},
\{p_1,p_9\})$. The blue/bend/thick/horizontal edges are \emph{communication
  edges}. All other edges are \emph{migration edges}.
\label{fig:time_expanding_graph}}
\end{figure}
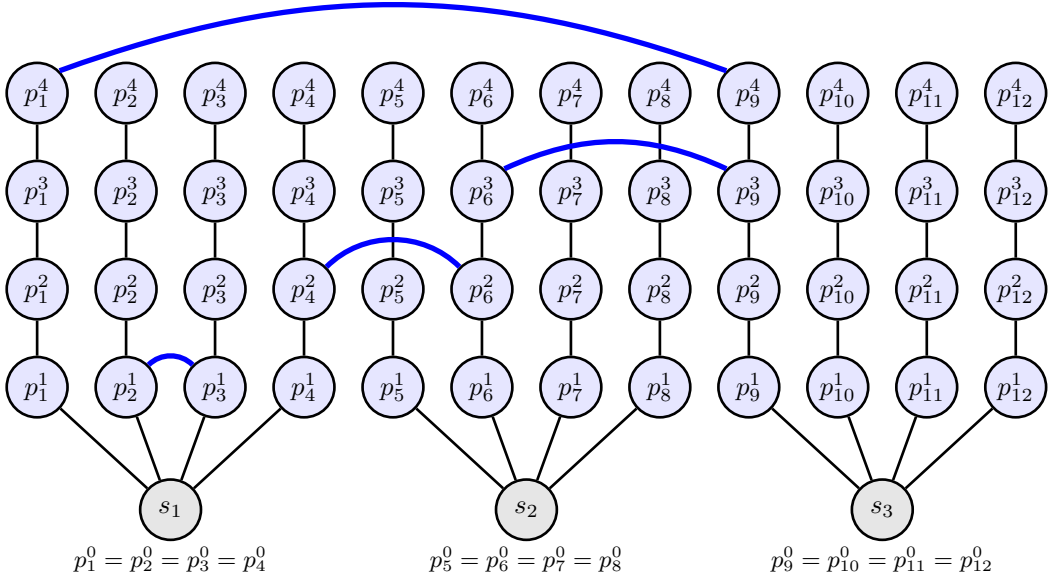

\subsection{Linear Program Formulation}
Suppose we assign edge-lengths $d_e \in \{0, 1\}$ to the edges of the graph
$\TG_\sigma$, with the following meaning: if a communication edge 
$(p^t,q^t)$ has length $1$, processes $p$ and $q$ \emph{may} be scheduled on
different servers at time $t$ (then the request is served via communication).
Similarly, if a migration edge $(p^{\tau-1},p^\tau)$ has length $1$ the 
process $p$ may switch its server at time $\tau$ (meaning servers for time $\tau-1$
and $\tau$ are different).

With this interpretation a distance of $0$ between a node $p^\tau$ and a server
node $s$ implies that $p$ is scheduled on $s$ at time $\tau$. Consequently, the
distance between any two server nodes $s$ and $s'$, must be at least 1, as
otherwise the process time nodes on a shortest path between $s$ and
$s'$ would have to be scheduled on both servers.

\def\calP{\mathcal{P}}
\def\lpprimal{LP~\texttt{Primal}}
\def\lpdual{LP~\texttt{Dual}}
Let $\calP_S$ denote the set of simple paths in $\TG_\sigma$ that connect two 
different server
nodes. The above discussion implies that \lpprimal\/ in Figure~\ref{lp:lp1}
gives a lower bound on $\OPT(\sigma)$.

\begin{figure}[h]
\hfill
\begin{minipage}{0.3\textwidth}
\begin{align*}
    \text{\underline{\lpprimal:}} \\[10pt]
	\min&&\quad\sum_{e \in E_{\TG_\sigma}} d_e\\
	\forall  P \in \mathcal{P}_S:&&  \quad\sum_{e \in P} d_e &\ge 1\\
    \forall  e \in E_{\TG_\sigma}:&&  \quad d_e &\ge 0\\
\end{align*}
\end{minipage}
\hfill\hfill
\begin{minipage}{0.3\textwidth}
\begin{align*}
    \text{\underline{\lpdual:}} \\[10pt]
	\max&&\quad\sum_{P \in \calP_S} f_{P}\\
	\forall  e \in E_{\TG_\sigma}:&&  \quad\sum_{P: e \in \calP_S} f_{P} &\le 1\\
    \forall  P \in \calP_S:&&  \quad f_{P} &\ge 0\\
\end{align*}
\end{minipage}
\hfill\hspace*{0pt}
\caption{Linear program relaxation of the generalized learning problem}
\label{lp:lp1}
\end{figure}

We can interpret the dual linear program $\texttt{Dual}$ in \ref{lp:lp1} as a
flow problem. We define $c_e := \sum_{P \in \calP_S: e \in P} f_{P}$ as the
\emph{congestion} of an edge $e$. 
Furthermore, we can interpret the variable $f_{P}$ as the flow we send from one
server to another on the path $P$. Therefore we call $\calP_S$ the set of
\emph{flow-augmenting paths} in the following.
The objective is to maximize the overall flow amount while ensuring that the
congestion
of an edge is at most one.

\begin{lemma}
    Let $\mathcal{P} \subseteq \mathcal{P}_S$ be a set of edge-disjoint simple
    paths. Then, $|\mathcal{P}| \le \OPT(\sigma)$.
\label{lemma:edge_disjoint}
\end{lemma}
\begin{proof}
We set $f_{P}=1$ for $P \in \mathcal{P}$, i.e., we send a flow amount of one on
each such path. Since the paths are edge-disjoint, the congestion on an edge is
at most 1. 
 Thus,
there is a feasible dual solution with objective value $|\mathcal{P}|$.
Hence, the lemma follows from weak duality.
\end{proof}

\subsection{Clustering}

We separate our online algorithms into two parts, a \emph{clustering part} and
a \emph{scheduling part}. The clustering part maintains a partition of the
process set $V$ into disjoint pieces $C_1,C_2,\dots$. We call this partition a
\emph{clustering} and the individual pieces \emph{clusters}. The scheduling
part of an algorithm then decides which cluster is going to which server but it
does not split clusters among different servers. In the following we describe
the flow clustering procedure, which operates on a subset $S\subseteq V$ of
processes and decides how to split these subsets into clusters. This procedure
forms a basic building block in our algorithms and is used for making
clustering decisions when the optimum cost is rather low.

\mathversion{bold}
\paragraph{Flow Clustering Procedure}\label{flow_clustering_procedure}
\mathversion{normal}

\def\calS{\mathcal{S}} 
Let $S\subseteq V$ denote a subset of processes. The flow procedure $\FLOW$ receives 
requests among processes in $S$ and maintains a clustering $\calS$ of $S$. The
procedure stops when $\OPT$ must have had a large cost \emph{for
  requests in $S$}.\footnote{For now we assume that their are no requests
of the form $\{s,x\}$, with $s\in S$ and $x\notin S$. Later in 
Section~\ref{sec:upper_bound_two_eps}, we will have such requests. Then we will \emph{merge} 
clustering procedures.}

Recall that we associate each server with a unique color. In addition to
maintaining the clustering the flow procedure also assigns colors to clusters
with the understanding that (usually) a cluster $C\in\calS$ with color $c$ is
scheduled on the server with color $c$. However, the scheduling part may
deviate from this \enquote{suggestion} in order to enforce balance constraint.
Moreover, not every cluster has a color. There is one \emph{special cluster}
that is not assigned a color. All other clusters in $\calS$ have a color
assigned and are called \emph{colored clusters}.

Each request between processes in $S$ might trigger an update of $\calS$. 
We (artificially) define the cost of the flow procedure for an update operation
as the number of processes that change their color. Later we will
show that the actual cost for scheduling the clusters is not too much more then
this cost.

% linked and free processes
The procedure distinguishes between two types of processes in $S$: \emph{linked}
processes $L$ that always stay within their initial cluster, and \emph{free}
processes $F$, which can be migrated between clusters. 
Initially, all processes in $S$ are linked. 
During the procedure, linked processes may become free, but
once free, a process may not become linked again. Crucially, the number of free
processes (roughly) acts as a lower bound of an optimum solution \emph{for
  requests involving processes in $S$}.
The procedure is started with a cost parameter $Z$. Once the number of free 
processes reaches $Z$ the procedure stops. Further requests involving a process
from $S$ are handled by another building block.

Let $\TG[S]$ be the time-expanded graph induced by the process set $S$. To
compare the update cost of the flow procedure to $\OPT$ the procedure maintains
a set of disjoint paths $\mathcal{P}\subseteq\calP_S$ in $\TG[S]$. By setting
$f_{P} = 1$ for $P \in \mathcal{P}$, we get a feasible dual solution to the
linear program.
In order to construct these paths we will
prove the following invariant:
\begin{invariant}
Let $\tau$ be the current time-step. For each $c$-colored cluster $C$ and each 
process $p \in C$, there is a path $Q_\tau(p)$,
in the time-expanded graph $\TG[S]$, that
\begin{enumerate}[i)]
    \item connects the server node with color $c$ to $p^\tau$ (and does not
    contain other server nodes); \label{invA}
    \item only contains nodes of the form $v^{i}$, $i\le\tau$, $v\in C$; \label{invB}
    \item exactly one of the processes that corresponds to nodes on the
    path is linked; \label{invC}
    \item is edge disjoint from all paths in $\mathcal{P}$.\label{invD}
\end{enumerate}
\label{inv1}
\end{invariant}

\subsubsection{Request Processing}
The procedure operates in iterations. 
At the beginning of an iteration, the free processes are part of the special
cluster.
A request $\sigma_t=(p,q)$ between processes $p$ and $q$ with
$p,q\in S$ is handled as follows.
\begin{enumerate}
\item If $p$ and $q$ are in the same cluster, we ignore the request as it does
not require any communication/migration. 
\item Suppose, exactly one of $p$ and $q$ is in the special cluster.
W.l.o.g., let $q$ be the process in the special cluster. Then, $p$ must
belong to a colored cluster $C$. We migrate $q$ to $C$. 
\item Otherwise, both $p$ and $q$ belong to colored clusters.
According to Lemma~\ref{lemma:invariant} there are paths $Q_t(p)$ and $Q_t(q)$,
which fulfill the properties of Invariant~\ref{inv1}. In particular, there are
exactly two linked processes $p'$ and $q'$ that are contained in the paths
$Q_t(p)$ and $Q_t(q)$, respectively. 
\begin{itemize}
\item We turn $p'$ and $q'$ into free processes. 
\item We update the set $\mathcal{P}$ by adding the path
$P=Q_t(p)\circ (p^t, q^t) \circ Q_t^{-1}(q)$.
\item We move all free processes to the special cluster.
\item If now $|F| \ge Z$, the procedure terminates. Otherwise, we proceed to the
next iteration.
\end{itemize}
\end{enumerate}

\subsubsection{Correctness Analysis}
\begin{lemma}
    The Invariant~\ref{inv1} is maintained throughout the execution of the procedure.
    \label{lemma:invariant}
\end{lemma}
\begin{proof}
  Fix a colored cluster $C$ and a process $p\in C$. 
  We first show that the invariant holds at the start of an iteration.

  At this
  time all processes in $C$ are linked, since the free processes are
  moved to the special cluster at the end of the previous iteration. Hence, $p$
  is linked. Whenever,
  a path is added to $\calP$ at the end of an iteration all linked processes that correspond to nodes
  on the path are converted into free processes (these are exactly two). Hence,
  the paths in $\calP$ cannot contain any node of the form $q^i$, $i\ge 1$
  where $q$ is linked (they can contain $q^0$ as this is a server node and
  exists under many different names).

  Therefore, the path $p^0,p^1,\dots,p^\tau$ is edge-disjoint from any path in
  $\calP$. All other constraints in Invariant~\ref{inv1} trivially hold for
  this path. This completes the proof for the start of an iteration.

  Now, suppose that the time-step $\tau$ increases and 
  a request $\sigma_\tau=\{x,y\}$ occurs (for the new value of $\tau$).
  Let $p$ denote a process in a colored
  cluster $C$ and let $s_c$ denote the server node with color $c$. By
  induction hypothesis there is a path $P$ from $s_c$ to $p^{\tau-1}$ that fulfills
  constraints~\ref{invB},\ref{invC}, and \ref{invD} of Invariant~\ref{inv1}. By
  extending $P$ by the edge $p^{\tau-1},p^{\tau}$ we get the desired path for
  $p^\tau$.

  It remains to construct a path for the case that a colored cluster $C$ changes
  by adding a free process to it. For this suppose that for the request
  $\sigma^\tau=\{x,y\}$, $y\in C$ and $x$ is a free process.
  Then, we already have constructed a path that connects $s_c$ to $y_\tau$. By extending this path
  via the edge $(x^\tau,y^\tau)$ we obtain the desired path to $x^\tau$.
  Clearly, the edge $(x^\tau,y^\tau)$ cannot be contained in any path from
  $\calP$ as these paths only contain nodes with time-stamp strictly less
  than $\tau$.
\end{proof}

\begin{lemma}
    The paths in $\mathcal{P}$ are pairwise edge disjoint.
\end{lemma}
\begin{proof}
We prove the lemma by induction. Initially, the set $\calP$ is empty.
Now, assume that the paths in the current set $\calP$ are pairwise
disjoint and that the request $\sigma_\tau=\{p, q\}$ ends the current iteration,
i.e., a new path $P=Q_\tau(p)\circ(p^\tau, v^\tau)\circ Q_\tau^{-1}(q)$ is added to $\calP$.

By Lemma~\ref{lemma:invariant}, the paths $Q_\tau(u)$ and $Q_\tau(v)$ are edge
disjoint with any path in $\calP$. Clearly, no path in $\calP$ can contain the
edge $(p^\tau, q^\tau)$. Furthermore, the paths $Q_\tau(p)$ and $Q_\tau(q)$
cannot share edges, since $p$ and $q$ are in different colored clusters
(Property~\ref{invB} of Invariant~\ref{inv1}). Hence,
the path $P$ is a simple path from one server node to another, that is edge
disjoint with any path in $\mathcal{P}$.
\end{proof}

\subsubsection{Cost Analysis}
Let $F$ be the current set of free processes of the $\FLOW$-procedure, and let $X$ be its special cluster.

\begin{lemma}
     $|F| = 2|\mathcal{P}|$.
    \label{lemma:flow_dual}
\end{lemma}
\begin{proof}
    Initially, the sets $F$ and $\mathcal{P}$ are both empty. We change them only at the end of each iteration. We add two new processes to the set $F$, and one new path to $\mathcal{P}$. Hence, at all times $|F| = 2|\mathcal{P}|$ holds.
\end{proof}

\begin{lemma}
     At any time the update cost of the $\FLOW$ procedure is at most $|F|^2 - |X|$.
     \label{lemma:flow_cost_F2}
\end{lemma}
\begin{proof}
We prove the statement by using induction on the number of requests. Initially,
$F$ and $X$ are both empty. Hence, the base case is trivial. For the induction
step assume, that prior to the current request the cost is at most
$|F|^2 - |X|$.

There are two cases where the procedure incurs cost and/or the sets $F$ or $X$ change.

\begin{itemize}
\item A request $\{x,p\}$ occurs between a process $x\in X$ and a process $p$ from a
colored cluster $C$:\\
Subsequently, $x$ migrates to the colored cluster, which induces an update cost
of $1$. Since $x$ leaves the special cluster $X$, its size decreases by one.
Hence, the cost after this step is at most
    \[|F|^2 - |X| + 1 = |F|^2 - (|X| - 1) = |F|^2 -|X'|\enspace,\] 
    where $X'$ is the new special cluster. The set $F$ does not change.
    \item A request $\{p,q\}$ occurs between processes $p\in C_p$ and $q\in C_q$, where $C_p,C_q$ are colored clusters:\\
    In this case we create two new free processes and all free processes are
    moved to the special cluster. Let $F'$ be the new set of free processes and
    $X'$ the special cluster after this step. By construction, $X' = F'$. The
    cost of migration of the free processes is at most $|F'|$.
    Hence, the new update cost is at most
    \begin{align*}
    |F|^2 - |X| + |F'|  = (|F'|-2)^2 - |X| + |F'| 
                = |F'|^2 - 4 |F'| + 4 - |X| + |F'| 
                \le |F'|^2 - |X'|.
    \end{align*}
    The last inequality holds because $|F'| \ge 2$. Hence, the invariant is maintained.
\end{itemize}
\end{proof}

\mathversion{bold}
\section{A Deterministic $\mcO(\max(\sqrt{k\ell \log k}, \ell \log k))$-competitive Algorithm with Augmentation $1+\epsilon$}
\mathversion{normal} In this section, we introduce a deterministic algorithm
that solves the generalized learning problem online with a competitive ratio of
$\mcO(\max(\sqrt{k\ell \log k}, \ell \log k))$ and augmentation $1+\epsilon$.
%The algorithm is designed to compete against the dynamic optimal algorithm, and, hence, it is also competitive against the static optimal algorithm.

A key ingredient of our approach is the deterministic online algorithm for the
standard learning model presented in \cite{soda21}, which we use as a
subroutine $\SODA$.

%For completeness, we restate the following result:
%\begin{theorem}[\cite{soda21}]\ruslan{the theorem is not stated directly like this in soda paper}
%    There exists a deterministic online algorithm $\SODA$ for the standard learning model with a competitive ratio $\mcO(\ell \log k)$ and a  $1+\epsilon$ augmentation.\ruslan{do we need to have $\epsilon$ in the O notation?}
%\end{theorem}

\subsection{Algorithm}
Let $\mathcal{I}$ denote the initial configuration, where each set
$C \in \mathcal{I}$ corresponds to the processes initially located at
some server.
We assign each $C \in \mathcal{I}$ the color of its respective server.
The algorithm for the general learning problem operates in two phases and
consists of two main components: the $\FLOW$ procedure (\ref{flow_clustering_procedure}) and the $\SODA$
subroutine.
\begin{enumerate}
\item Phase 1: \\We execute the flow procedure with cost parameter
$Z=\min(\sqrt{k\ell \log k}, \epsilon k)$ on the process set $V$, starting from the
initial configuration $\mathcal{I}$, until it terminates. After processing each
request, the flow procedure outputs a clustering $\mathcal{S}$, which
consists of $|\mathcal{I}|$ colored clusters and one special cluster. The
assignment of clusters to servers is as follows: the colored clusters are
scheduled on their corresponding servers, while the special cluster is placed
on Server~1.
\item Phase 2: \\Once the flow procedure terminates, we reset the configuration
to the initial state $\mathcal{I}$. We then reprocess the request sequence from
the beginning using the $\SODA$ algorithm.
\end{enumerate}
Note that $\SODA$ operates in the standard learning model, which means that it assigns
entire connected components of the demand graph to a single server. However, we
will demonstrate that the cost incured during the second phase is not
significantly higher than that of the first phase.
%The complete algorithm is
%presented in Algorithm~\ref{alg:general_one_eps}. \ruslan{do we need pseudocode
%  for this?}

%\begin{algorithm}[t]
%\caption{Algorithm with $1+\epsilon$ augmentation}\label{alg:general_one_eps}
%\SetKwInOut{KwIn}{Input}
%% \SetKwInOut{KwOut}{Output}
%\KwIn{resource augmentation $\epsilon$, initial scheduling $\mathcal{I}$}
%% \KwOut{}
%\DontPrintSemicolon
% 
%execute procedure $\FLOW(\min(\sqrt{k\ell \log k}, \epsilon k)$ until it terminates\;
%%$t^* \gets$ time step of termination\;
%revert the current scheduling to  $\mathcal{I}$\;
%start algorithm $\SODA$ from the beginning of the request sequence\;
%\end{algorithm}

\subsection{Correctness Analysis}

\begin{lemma}
    The algorithm uses augmentation at most $1+\epsilon$.
    \label{lemma:feas_one_eps}
\end{lemma}
\begin{proof}
By construction of the flow procedure, only free processes are allowed to
leave their initial cluster. The procedure ends if the number of free
processes reaches $Z=\min(\sqrt{k\ell \log k}, \epsilon k)$. Therefore, at most
$\min(\sqrt{k\ell \log k}, \epsilon k) \le \epsilon k$ processes will leave
their initial cluster during the execution of the
flow procedure. Since each 
cluster uniquely corresponds to a server, at most
$\epsilon k$ processes will leave their home server. Consequently, no server can
be overloaded by more than $\epsilon k$ processes. In the second phase, we apply
the $\SODA$ algorithm from~\cite{soda21}, which by design, uses augmentation at most $1+\epsilon$.
\end{proof}

\subsection{Cost Analysis}
The following lemma is not explicitly stated in~\cite{soda21} but directly
follows from their analysis. See Appendix~\ref{appendix_a}.
\begin{restatable}{lemma}{LemmaSodaCost}[\cite{soda21}]
    For any request sequence $\sigma$ that obeys the learning restriction
    the cost $\SODA(\sigma)$ is at most $\mcO(k\ell\log k)$.
    \label{lemma:soda_cost}
\end{restatable}

Next, we show that the competitive ratio of the online algorithm is $\mcO(\max(\sqrt{k\ell \log k}, \ell \log k))$.
\begin{lemma}
    For a request sequence $\sigma$ the 
    cost of the online algorithm is at most $\mcO(\max(\sqrt{k\ell \log k}, \ell \log k))\cdot\OPT(\sigma)$. 
    \label{lemma:cost_one_eps}
\end{lemma}
\begin{proof}
According to Lemma~\ref{lemma:flow_cost_F2}, at any point during the execution
of the flow procedure, its cost is at most
$|F|^2-|X|$. This gives
\begin{equation*}
\FLOW(\sigma)\le |F|^2-|X|\le 2|\calP|\cdot Z\le 2Z\cdot\OPT(\sigma)\enspace,
\end{equation*}
where the second inequality holds due to Lemma~\ref{lemma:flow_dual} and the fact
that $|F|\le Z$ as the procedure terminates when $|F|$ reaches
$Z$.
%\harry{strictly this requires that $Z$ is even...}
The final inequality
follows because the paths in $\calP$ form a feasible dual solution to the
linear program, and, hence, $\OPT(\sigma)\ge|\calP|$.

The cost for reverting the configuration to the initial state is at most the
cost incurred in the first phase, i.e., at most $\FLOW(\sigma)$.

The cost of the second phase is at most the cost of the $\SODA$-algorithm.
According to Lemma~\ref{lemma:soda_cost} this is at most $ck\ell\log k$ for some constant $c$.
When the flow procedure ends $\OPT(\sigma)\ge|\calP|\ge Z/2$.
We now consider two cases:
\begin{itemize}
    \item $\sqrt{k\ell \log k} \le \epsilon k$. Then
    \[
    \SODA(\sigma) = ck\ell\log k = c\sqrt{k\ell\log k} \sqrt{k\ell\log k}  \le 2c \sqrt{k\ell\log k}\cdot\OPT(\sigma)\enspace,
    \]
    because $\OPT(\sigma)\ge \sqrt{k\ell\log k}/2$.
    Therefore, the total cost for both phases is at most
    \[
    4Z\cdot\OPT(\sigma) + 2c\sqrt{k\ell\log k}\cdot\OPT(\sigma) = \mcO(\sqrt{k \ell  \log k})\cdot\OPT(\sigma)\enspace.\]
    \item $\epsilon k < \sqrt{k\ell \log k}$. In this case we have
    \[
    \SODA(\sigma) = ck\ell\log k= \frac{c}{\epsilon} \ell\log(k)\cdot\epsilon k   \le  \frac{2c}{\epsilon} \ell  \log(k)\cdot\OPT(\sigma)\enspace.
    \]
    Therefore the total cost for both phases is at most 
    \[
    4Z\cdot\OPT(\sigma) + \frac{2c}{\epsilon} \ell  \log(k)\cdot\OPT(\sigma) \le \mcO(\sqrt{k \ell  \log k} + \ell \log k)\cdot\OPT(\sigma).
    \]
\end{itemize}
Combining both cases gives the lemma.   
\end{proof}

\noindent
Lemmas \ref{lemma:cost_one_eps} and \ref{lemma:feas_one_eps} give the following
theorem.
\UpperBoundsOneEps*

\mathversion{bold}
\section{A Deterministic $\mcO(\sqrt{k})$-competitive Algorithm with Augmentation $2+\epsilon$}\label{sec:upper_bound_two_eps}
\mathversion{normal}

In this section, we introduce an online algorithm that solves the generalized learning problem with augmentation $2+\epsilon$ and achieves a competitive ratio of $\mathcal{O}(\sqrt{k})$.
We build on the concepts discussed in the previous sections, particularly the $\FLOW$ procedure. However, since our goal is to achieve a $\mcO(\sqrt{k})$ competitive ratio, we cannot simply apply the procedure as before. 
%Thus, some modifications are necessary. The algorithm is competitive against the dynamic optimal algorithm, and therefore also competitive against the static optimal algorithm.

\subsection{Overall Algorithm}

The algorithm consists of two main components: The \emph{Clustering algorithm} and the \emph{Scheduling algorithm}. 
\begin{enumerate}
    \item The \emph{Clustering algorithm} takes the request sequence as input and maintains a clustering $\mathcal{S}$ of the processes. When a new request is received, the algorithm may adjust the clustering, causing the processes to change their clusters. 
    \item The \emph{Scheduling algorithm} takes a clustering $\mathcal{S}$ from the clustering algorithm and assigns the clusters to the servers. Changes in clustering may require some clusters to be rescheduled. Each cluster is always entirely scheduled on a single server.
\end{enumerate}

\subsection{Clustering Algorithm}
The basic idea of the clustering algorithm is to group processes into a set of clusters $\mathcal{S}$, by running multiple $\FLOW$ procedures (see Section~\ref{flow_clustering_procedure}), each on a connected component of the demand graph $\bar{G}_{\sigma}$. 

\subsubsection{Clustering Rules}
% \paragraph{Clustering rules}
Let $\mathcal{C}$ be the set of connected components of the current demand graph $G_{\sigma}$. We partition $\mathcal{C}$ into two sets $\mathcal{C}_S$ and $\mathcal{C}_L$. 

\begin{itemize}
    \item For each $C \in \mathcal{C}_L$ we have that $C \in \mathcal{S}$, that means that $C$ forms itself a cluster. We call $C$ a $\emph{large cluster}$.
    \item For each $C \in \mathcal{C}_S$ we maintain a $\FLOW$ procedure $I(C)$ with  
    cost parameter $Z = \sqrt{k}$, that operates on the connected component $C$. We call $I(C)$ a \emph{procedure} or \emph{instance} and say that $I(C)$ \emph{manages} the component $C$. 
    Each procedure $I(C)$ partitions its connected component $C$ into smaller sub-clusters $\mathcal{S}_C$. To avoid ambiguities, we call them \emph{pieces}. A piece is either colored, or it is the special piece of its corresponding $\FLOW$ procedure. Furthermore, each procedure maintains a set of free processes $F_{C} \subseteq C$.

    \paragraph{Assignment of pieces to clusters}
    We say a piece $S$ is \emph{$\delta$-monochromatic}, if at least $\delta|S|$ ($\delta \ge \frac{1}{2}$) processes in $S$ share the same $c$-colored home server. The remaining processes in $S$ are free.
    %We call $c$ the \emph{majority color} of $S$.
    For every color $c$ the clustering algorithm maintains a special cluster $S_c$,
    which we call the \emph{$c$-colored cluster}. A piece $S$ either forms a \emph{singleton cluster} that only contains piece $S$, or it is assigned to the $c$-colored cluster, if $c$ is the color of $S$. 
    The assignment of pieces to clusters is done as follows.
    
    % \begin{itemize}
    % \item 
    All pieces initially consist of a single process, and thus belong to their corresponding $c$-colored cluster.
    % \item 
    Whenever a piece $S$ changes we examine the new piece $S'$. Let $c$ be the color of $S'$.
    \begin{itemize}
    \item
    If $S'$ is not $\frac{1}{2}$-monochromatic, it becomes a singleton cluster.
    \item
    If $S'$ is $\frac{3}{4}$-monochromatic, it is assigned to the $c$-colored cluster.
    \item Otherwise, $S'$ is assigned to the $c$-colored cluster iff $S$ was assigned to this cluster. 
    This means, that once a piece is assigned to a colored cluster, it will only be reassigned to a singleton cluster after a significant number of free processes have migrated to that piece.
    %This means that if $S$ was assigned to the color $c$ cluster we assign $S'$ to it.
    %Otherwise, $S'$ forms a singleton cluster.
    \end{itemize}
    % \end{itemize}
    
\end{itemize}

\paragraph{Cost}
The cost of the clustering algorithm is defined as the total number of cluster changes made by processes. Each time a process is reassigned to a different cluster, it incurs a cost of 1.

During the execution of the algorithm, sets of processes are often \emph{merged}. Let $A$ and $B$ be two disjoint sets of processes. When merging them into a single set $C = A \cup B$, the algorithm migrates all processes from the smaller set to the larger one. Afterward, the (now) empty set is deleted.

Notice, not all merges incur a cost. If both $A$ and $B$ are $c$-colored pieces that already reside within the same colored cluster, merging them does not involve any cost, as no process changes its cluster. On the other hand, if either $A$ or $B$ is a singleton cluster, the merge will incur a cost equivalent to $\min(|A|, |B|)$, reflecting the number of processes that need to change clusters.

\subsubsection{Instance Management}

\paragraph{Instance initialization}
The clustering algorithm starts by creating an instance $I(\{v\})$ for each process $v \in V$. Each set $\{v\}$ is assigned the color of $v$'s home server. 
As a result, each process is initially contained in the $c$-colored cluster, if $c$ is the color of its home server.
Thus, initially $\mathcal{C}_S = \{\{v\}~|~ v \in V\}$ and $\mathcal{C}_L = \emptyset$. 

\paragraph{Instance merge}
Throughout the execution of the algorithm, whenever a request involves two different components $A, B \in \mathcal{C}_S$, we merge their corresponding instances, $I(A)$ and $I(B)$, into a single instance $I(C)$, where $C = A \cup B$. 
Each instance maintains at most one $c$-colored piece and exactly one special piece. This invariant must be preserved during the merge.
The merging process is implemented as follows:
\begin{itemize}
    \item Colored pieces: For each color $c$, if there is a $c$-colored piece in both $A$ and $B$ (i.e., $C_A \in \mathcal{S}_A$ and $C_B \in \mathcal{S}_B$), these pieces are combined into a single piece $C_A \cup C_B$ in the new instance $I(C)$.
    \item Special pieces: Let $X_A$ and $X_B$ be the special pieces of instances $I(A)$ and $I(B)$, respectively.
These are combined into a single special piece $X_C = X_A \cup X_B$. 
    \item 
    Flow-augmenting paths and free processes: The set of flow-augmenting paths for the merged instance is set as $\mathcal{P}_C = \mathcal{P}_A \cup \mathcal{P}_B$, combining the paths from both original instances.
Similarly, the set of free processes is set as $F_C = F_A \cup F_B$, uniting the free processes from $A$ and $B$.
\end{itemize}

\paragraph{Instance stop}
A procedure continues to operate until the number of free processes in its component reaches $\sqrt{k}$, at which point the procedure stops.
If a procedure $I(C)$ stops, all its pieces (the whole connected component $C$) are merged into a singleton cluster. Consequently, $C$ moves from $\mathcal{C}_S$ to $\mathcal{C}_L$. 
%We call an instance $I(C)$ with $C \in \mathcal{C}_S$ an \emph{active} instance. All other instances are terminated instances.

\subsubsection{Request Processing}
%\paragraph{Request processing}
The requests are processed as follows. 
Let $\sigma_t=\{u, v\}$ be the current request. 

If $u$ and $v$ belong to different connected components $A$ and $B$, respectively, we first execute the component merge operation \texttt{CompMerge($A, B$)}. 
After this operation, it is guaranteed that $u$ and $v$ are part of the same connected component $C$. 

If $C \in \mathcal{C}_L$, the request is served without any cost, since the algorithm ensures that $u$ and $v$ are in the same cluster. Otherwise, if $C \in \mathcal{C}_S$, the request is forwarded to the corresponding procedure $I(C)$ for further processing. 

\paragraph{\texttt{CompMerge} operation }
A component merge operation \texttt{CompMerge($A, B$)} is performed as follows. Let $\mathcal{C}'_S$ and $\mathcal{C}'_L$ be the connected component sets after the \texttt{CompMerge} operation.

\begin{enumerate}
    \item $A \in \mathcal{C}_L$ or $B \in \mathcal{C}_L$\\
    Without loss of generality, assume $A \in \mathcal{C}_L$. 
    % \begin{itemize}
        % \item 
        If $B$ is also a large cluster, we merge the two clusters. Thus, the new large component set becomes
        $\mathcal{C}'_L = \mathcal{C}_L - \{B, A\} + \{A\cup B\}$.
        
        % \item 
        In case $B \in \mathcal{C}_S$, we stop the $I(B)$ instance that manages $B$, and merge all pieces of $I(B)$ (i.e., the entire component $B$) into cluster $A$. As a result, 
    $\mathcal{C}'_S = \mathcal{C}_S - \{B\}$ and $\mathcal{C}'_L = \mathcal{C}_L - \{A\} + \{A\cup B\}$.
    % \end{itemize}
 
    \item $A \in \mathcal{C}_S$ and $B \in \mathcal{C}_S$\\
    In this case, there are instances $I({A})$ and $I({B})$ that manage $A$ and $B$, respectively.
    We merge $I({A})$ and $I({B})$ into a single instance $I(A \cup B)$. 
    The component sets are updated as follows: $\mathcal{C}'_S = \mathcal{C}_S - \{B, A\} + \{A\cup B\}$.
\end{enumerate}

\subsection{Scheduling Algorithm}
The scheduling algorithm is taken directly from \cite{racke2023polylog}. For completeness, we repeat the algorithm here. 

The scheduling algorithm receives a set of clusters $\mathcal{S}$ from the clustering algorithm and maintains an assignment of the clusters to the servers in a way that keeps the load on each server within specified bounds.
The scheduling algorithm assigns each $c$-colored cluster to the server $s$ with the matching color $c$. All other clusters are scheduled on arbitrary servers.

Upon a new request the clustering algorithm could change existing clusters. Some
processes might change their cluster, and some clusters might merge. This changes the size of clusters and, hence, the load distribution among the servers can become imbalanced.
The scheduling procedure re-balances the distribution of clusters among servers such that there are at most $(2 + \epsilon)k$ processes on any server, for an $\epsilon > 0$.

Assume, that after the execution of the clustering algorithm server $s$ has load greater
than $(2 + \epsilon)k$. We perform the following re-balancing procedure. While server $s$ has load greater than $2k$, we take a non-colored cluster $C$ in $s$ and move it to a server $s'$ with load at most $k$. 
Such a cluster $C$ must exist, since the overall number of processes in colored clusters is at most $2k$.
Furthermore, such a server $s'$ must also exist because the average load is at most $k$. Since $|C| \le k$, $s'$ has a load at most $2k$ afterwards.

\subsection{Analysis}
For each $C \in  \mathcal{C}_S$, the procedure $I(C)$ maintains two sets: the set of flow-augmenting paths $\mathcal{P}_{C}$ and the set of free processes $F_{C}$. 
By construction of the $\FLOW$ procedure, each path $P \in \mathcal{P}_{C}$ contains only nodes of the form $v^i$ for $i\ge 0$ and $v \in C$. Additionally, $F_{C}$ is always a subset of $C$.

We extend these definitions to all components in $\mathcal{C}$. We define:
\begin{itemize}
    \item $\mathcal{P}_{C}$: The set of flow-augmenting paths of procedures that have at any point managed $C$ or any subset of $C$.
    \item $F$: The set of free processes across all procedures that have operated at any point in time.
    \item $F_C = F \cap C$, for a $C \subseteq V$.
\end{itemize}

Due to the design of the instance merge operation, for any $C \in \mathcal{C}_S$, these definitions naturally align with the sets of flow-augmenting paths and free processes maintained by the instance $I(C)$.

Additionally, for each $C \in \mathcal{C}_S$, let $X_C$ denote the special piece of the instance $I(C)$. 

Finally, let $\mathcal{P} = \bigcup_{C \in \mathcal{C}} \mathcal{P}_C$ represent the set of all flow-augmenting paths across all components. By Lemma~\ref{lemma:flow_dual} and by construction of the instance/component merge operation, $|\mathcal{P}| = \bigcup_{C \in \mathcal{C}} |\mathcal{P}_C| = \bigcup_{C \in \mathcal{C}} 2|F_C| = 2 |F|$

\subsubsection{Correctness Analysis}

First, we prove the augmentation guarantees of the algorithm presented in this section.

\begin{lemma}
    The algorithm uses at most augmentation $2+\epsilon$.
    \label{lemma:feas_two_eps}
\end{lemma}
\begin{proof}
    Let $C$ be a colored cluster. By construction, at most half of the processes in $C$ are free. Since there are $k$ processes that have a $c$-colored home server, the overall size of a colored cluster is bounded by $2k$. 

    Each time the clustering algorithm modifies existing clusters and creates a load on a server $s$ with more that $2+\epsilon$ processes, the scheduling algorithm re-balances this server by moving non-colored clusters, such that it has at most $2k$ processes. Such a re-balancing is always possible, since the size of a colored cluster is at most $2k$. Thus, there is always a non-colored cluster that we can move away from $s$. This ensures that no server contains more than $2+\epsilon$ processes at any given time.
\end{proof}

Next, we show that the paths in $\mathcal{P}$ are pairwise edge-disjoint, which implies that there is a feasible dual solution to the linear program~\ref{lp:lp1} with cost $|\mathcal{P}|$. 
%In addition, we show that Algorithm~\ref{alg:general_two_eps} uses at most $2+\epsilon$ augmentation.
%By construction, for two components $A \neq B$, the paths $P_A \in\mathcal{P}_{A} $ and $P_B \in\mathcal{P}_{B}$ are edge disjoint.
%Let  $\mathcal{P} = \bigcup_{C \in \mathcal{C}_S} \mathcal{P}_{C}$. Thus,  $|\mathcal{P}| = \sum_{C \in \mathcal{C}} |\mathcal{P}_{C}|$. 

\begin{lemma}
    For two connected components $A, B\in \mathcal{C}$, the paths in $\mathcal{P}_{A} \cup \mathcal{P}_{B}$ are pairwise edge disjoint.
    \label{fact:comp_edge_disjoint}
\end{lemma}
\begin{proof}
By design of the $\FLOW$ procedure, a path $P_A \in \mathcal{P}_{A}$ contains only nodes of the form $v^i$ for $i\ge 0$ and $v \in A$, and a path $P_B \in \mathcal{P}_{B}$ contains only nodes of the form $u^i$ for $i\ge 0$ and $u \in B$. Since $A \cap 
B = \emptyset$, $P_A$ and $P_B$ cannot share an edge.
\end{proof}

\begin{lemma}
    The Invariant~\ref{inv1} is maintained after a procedure merge.
    \label{lemma:invariant_merge}
\end{lemma}
\begin{proof}
Assume, that Invariant~\ref{inv1} holds prior to a merge of two procedures $I(A)$ and $I(B)$ into procedure $I(C)$, $C = A \cup B$. We show that the invariant is preserved after the merge. 

W.l.o.g. let $v \in A$ be a process that belongs to a colored piece $J$ and let $Q_\tau(v)$ be the path that fulfills the properties of Invariant~\ref{inv1} prior to the merge. By construction, $v$ also belongs to some colored piece $J'$ after the merge. 
We show, that $Q_\tau(v)$ still fulfills the properties of Invariant~\ref{inv1}.
Properties i) and iii) still trivially hold, since the path remains unchanged.
By design, $J \subseteq J'$, thus, property ii) is also trivially maintained.

It remains to show that $Q_t(v)$ is edge disjoint with the flow-augmenting paths in $\mathcal{P}_{C} = \mathcal{P}_{A} \cup \mathcal{P}_{B}$. 
Since the invariant holds prior to the merge, $Q_t(v)$ is edge disjoint with paths in 
$\mathcal{P}_{A}$. Furthermore, by design, a path $P \in \mathcal{P}_{B}$ contains only nodes of the form $u^i$ for $i\ge 0$ and $u \in B$. The path $Q_t(v)$ contains only nodes of the form $w^i$ for $i\ge 0$ and $w \in A$. Hence, paths $Q_t(v)$ and $P$ cannot share an edge.
Therefore, the invariant remains valid.
\end{proof}

% \begin{lemma}
%     For each $C \in \mathcal{C}_S$, the flow-augmenting paths in $\mathcal{P}_C$ are pairwise edge disjoint.
% \label{lemma:edge_disjoint_CS}
% \end{lemma}
% \begin{proof} 
% We prove the statement by induction. Initially, $\mathcal{P}_C =\emptyset$ for each $C \in \mathcal{C}_S$. Now, assume that the paths in
% $\mathcal{P}_C$ are pairwise disjoint.

% There are two cases, where $\mathcal{C}_S$ or $\mathcal{P}_C$ changes:
% \begin{itemize}
%     \item 
% \end{itemize}
%     By Fact~\ref{fact:comp_edge_disjoint}, the paths in $\mathcal{P}_{A} \cup \mathcal{P}_{B}$ are edge disjoint, since $A \cap B = \emptyset$.

% According to Lemma~\ref{lemma:invariant} and \ref{lemma:invariant_merge} the Invariant~\ref{inv1} remains valid throughout the execution of a procedure, and after a merge of two procedures.

% Since Lemma~\ref{lemma:invariant} was proven inductively, it suffices to show that the invariant is still valid after the merge of two procedures.

% Now, since the Invariant~\ref{inv1} holds, we can apply the same arguments as in Lemma~\ref{} and show that the paths created by instance $I(C)$ are edge disjoint.\ruslan{too sketchy?}
% \end{proof}

\begin{lemma}
    For each $C \in \mathcal{C}$, the flow-augmenting paths in $\mathcal{P}_C$ are pairwise edge disjoint.
\label{lemma:edge_disjoint_CL}
\end{lemma}
\begin{proof} 
We prove the statement by induction.
Initially, $\mathcal{P}_C =\emptyset$ for each $C \in \mathcal{C}$, thus the base case is trivially fulfilled. 
For the induction step, assume, that prior to the current
request the statement holds for all $C \in \mathcal{C}$.

There are two cases, where $\mathcal{C}$ or $\mathcal{P}_C$ changes: the \texttt{CompMerge} operation, and if a procedure creates a new path.

\begin{itemize}
    \item Assume that the \texttt{CompMerge} operation merges two connected components $A$ and $B$. By induction hypothesis, the paths in $\mathcal{P}_{A}$ are pairwise edge disjoint. Same holds for the paths in $\mathcal{P}_{B}$. 
    By Fact~\ref{fact:comp_edge_disjoint}, the paths in $\mathcal{P}_{A} \cup \mathcal{P}_{B}$ are edge disjoint.
    
    \item Assume a new path $P$ is added to $\mathcal{P}_{A}$. Clearly, there is a procedure $I(A)$ that manages $A$.
    According to Lemma~\ref{lemma:invariant} and Lemma~\ref{lemma:invariant_merge}, the Invariant~\ref{inv1} remains valid throughout the execution of a procedure, as well as a merge of two procedures. Hence, the invariant is preserved at any time.
    Due to the inductive structure of the proof in Lemma~\ref{lemma:invariant}, we derive that the new path $P$ must be edge disjoint with paths in $\mathcal{P}_{A}$.
\end{itemize}
Hence, for each $C \in \mathcal{C}$, the flow-augmenting paths in $\mathcal{P}_C$ are pairwise edge disjoint.
\end{proof}

The previous lemmas lead to the following conclusion.

\begin{lemma}
    The flow-augmenting paths in $\mathcal{P}$ are pairwise edge disjoint.
    \label{lemma:edge_disjoint_all}
\end{lemma}
\begin{proof} 
According to Lemma~\ref{lemma:edge_disjoint_CL} for each $C \in \mathcal{C}$, the flow-augmenting paths in $\mathcal{P}_{C}$ are pairwise edge disjoint. 
By Lemma~\ref{fact:comp_edge_disjoint}, for two components $A \neq B$, the flow-augmenting paths in $\mathcal{P}_{A} \cup \mathcal{P}_{B}$ are also edge disjoint.
Hence, the paths in $\mathcal{P} = \bigcup_{C \in \mathcal{C}} \mathcal{P}_{C}$ must be pairwise edge disjoint.
\end{proof}

\newcommand{\CFlow}{\ensuremath{\operatorname{cost}_{\operatorname{flow}}}}
\newcommand{\CMerge}{\ensuremath{\operatorname{cost}_{\operatorname{merge}}}}
\newcommand{\CSmall}{\ensuremath{\operatorname{cost}_{\operatorname{merge}}}}
\newcommand{\CLarge}{\ensuremath{\operatorname{cost}_{\operatorname{large}}}}
\newcommand{\CSpecial}{\ensuremath{\operatorname{cost}_{\operatorname{special}}}}
\newcommand{\CColor}{\ensuremath{\operatorname{cost}_{\operatorname{color}}}}
\newcommand{\CMono}{\ensuremath{\operatorname{cost}_{\operatorname{mono}}}}
\newcommand{\chf}{\ensuremath{\operatorname{charge}_{\operatorname{f}}}}
\newcommand{\chm}{\ensuremath{\operatorname{charge}_{\operatorname{m}}}}

\newcommand{\CostCluster}{\ensuremath{\operatorname{cost}_{\operatorname{cluster}}}}
\newcommand{\CostSchedule}{\ensuremath{\operatorname{cost}_{\operatorname{schedule}}}}
\newcommand{\CostCom}{\ensuremath{\operatorname{cost}_{\operatorname{com}}}}
\newcommand{\CostMig}{\ensuremath{\operatorname{cost}_{\operatorname{mig}}}}

\subsection{Cost Analysis}
In this section, we analyze the costs incurred by the algorithm. By design, the algorithm avoids any communication cost.
Therefore, the only costs we need to consider are related to migration. 

A migration \emph{may} occur if a process changes its cluster or if a cluster is rescheduled to another server.
We denote the costs associated with these actions as $\CostCluster$ and $\CostSchedule$, respectively. Consequently, the overall cost incurred by the algorithm is at most $\CostCluster + \CostSchedule$.

We further decompose $\CostCluster$ into following subcategories:
\begin{itemize}
    \item $\CFlow$: The cost incurred by $\FLOW$ procedures.
    \item $\CLarge$: The cost of forming a large cluster, and merging a large cluster with other clusters.
    %the \texttt{CompMerge} operations, where one of the components is in $\mathcal{C}_L$.
    \item $\CSmall$: The cost of instance merge operations, which occur when two procedures merge. The cost involve merging special pieces, as well as pieces with the same color.
    \item $\CMono$: The cost of either moving a piece to a colored cluster, or splitting a piece out of its colored cluster.
\end{itemize}%\ruslan{explain the different costs here, more detailed?}
Thus, $\CostCluster = \CFlow + \CLarge + \CSmall + \CMono$.

\subsection*{Clustering Cost}
Now, we will analyze the different clustering costs in detail. 
As a first step, we derive a bound on $\CFlow$. 

\subsubsection*{Upper bound on $\CFlow$}
For a connected component $C \in \mathcal{C}$, let $\cost(C)$ denote the costs incurred by all procedures for requests $\sigma_t=\{u, v\}$, where both $u$ and $v$ are in $C$. 
%Let $X_C$ be the special cluster of a procedure $I(C)$. 

\begin{lemma}
For all components $C \in \mathcal{C}_S$, $\cost(C) \le |F_C|^2 - |X_C|$ holds.
\label{lemma:cost_instance_S}
\end{lemma}
\begin{proof}
We prove the statement by induction. Initially, the sets $F_C$ and $X_C$ are empty. Hence, the base case is trivial. For the induction step, assume, that prior to the current request $\cost(C) \le |F_C|^2 - |X_C|$ for all $C \in \mathcal{C}_S$  holds.

Due to the inductive nature of the proof of Lemma~\ref{lemma:flow_cost_F2}, it suffices to show that after a procedure merge, the invariant holds for the new procedure.

Assume that components $A, B \in  \mathcal{C}_S$ merge into $C = A \cup B$. We  show that the invariant holds for the new component $C$.
By construction, the set of free processes in $C$ is  $F_C= F_A \cup F_B$, and the special piece is $X_C = X_A \cup X_B$.
Clearly, $\cost(C) = \cost(A) + \cost(B)$. Therefore, 
\begin{align*}
        \cost(C)&= \cost(A) + \cost(B)\\ 
                &\le |F_A|^2 - |X_A| + |F_B|^2 - |X_B| \\
                &\le (|F_A| + |F_B|)^2 - (|X_A| + |X_B|) \\
                &= |F_C|^2 -|X_C|.
\end{align*}
The second inequality holds due to the fact that $x^2 + y^2 \le (x+y)^2$, for $x, y \ge 0$. Hence, the invariant is preserved for the new component $C$.
\end{proof}

%Next, we analyze the costs of components $C \in \mathcal{C}_L$:

\begin{lemma}
For all components $C \in \mathcal{C}_L$, $\cost(C) \le \sqrt{k}|F_C|$ holds.
\label{lemma:cost_instance_L}
\end{lemma}
\begin{proof}
We prove the statement by induction on the number of requests. Initially, the set $\mathcal{C}_L$ is empty, thus the base case trivially holds. 
For the induction step, assume that prior to the current request, $\cost(C) \le \sqrt{k}|F_C|$ holds for all $C \in \mathcal{C}_L$.

There are no instances that manage components $C \in \mathcal{C}_L$, and therefore no instance can generate new cost. However, there are two cases where the set $\mathcal{C}_L$ changes. We show that in each case, the invariant is preserved.

\begin{itemize}
    \item Component merge:\\
    Assume, the components $A$ and $B$ are merged into $C = A \cup B$, where w.l.o.g. $A \in  \mathcal{C}_L$. 
    Thus the set $ \mathcal{C}_L$ changes. We show that the invariant holds for the new component $C$.
    By construction, $F_C= F_A \cup F_B$ and $\cost(C) = \cost(A) + \cost(B)$. By induction hypothesis, $\cost(A) \le \sqrt{k}|F_A|$. Same holds for $B$, if $B \in \mathcal{C}_L$. Otherwise, by Lemma~\ref{lemma:cost_instance_S}, $\cost(B) \le |F_B|^2 -|S_B| \le \sqrt{k} |F_B|$, since by design, the number of free processes of a procedure is at most $\sqrt{k}$. Hence, 
    \begin{align*}
            \cost(C) &\le \sqrt{k}|F_A| +\sqrt{k} |F_B| 
                    %= \sqrt{k}(|F_A| + |F_B|) 
                    = \sqrt{k}|F_C|.
    \end{align*}
    Therefore, the invariant is preserved for the component $C$.
    \item Instance stop:\\
    By Lemma~\ref{lemma:cost_instance_S}, at the time an instance $I(C)$ stops, $\cost(C) \le |F_C|^2-|S_C| \le \sqrt{k} |F_C|$ holds.
    Subsequently, the component $C$ is moved to $\mathcal{C}_L$.
    Thus, the invariant is preserved.
    
\end{itemize}
\end{proof}

Applying the previous two lemmas gives:

\begin{lemma}
$\CFlow \le \sqrt{k} |F|$.
\label{lemma:costInst}
\end{lemma}
\begin{proof}
According to Lemmas~\ref{lemma:cost_instance_S} and ~\ref{lemma:cost_instance_L},
    \begin{align*}
        \CFlow  &= \sum_{C \in \mathcal{C}_S} \cost(C) + \sum_{C \in \mathcal{C}_L} \cost(C)\\ 
                &\le  \sum_{C \in\mathcal{C}_S} (|F_C|^2 - |S_C|) +  \sum_{C\in \mathcal{C}_L} \sqrt{k} |F_C| \\
                &\le \sqrt{k} \sum_{C \in \mathcal{C}} |F_C|\\
                &= \sqrt{k} |F|.
    \end{align*}
The second inequality holds because for each instance $I(C)$, $|F_C| \le \sqrt{k}$.
% The last inequality holds due to Lemma~\ref{}.  
\end{proof}

Next, we analyze $\CLarge$.
\subsubsection*{Upper bound on $\CLarge$}
\begin{lemma}
$\CLarge \le \sqrt{k}|F|$.
\label{lemma:costLarge}
\end{lemma}
\begin{proof}
We prove the statement using a charging argument. We distribute the cost among the free processes, and show that each free process is charged at most $\sqrt{k}$. We distinguish between two cases in which $\CLarge$ is incurred:

\begin{itemize}
    \item Formation of a large cluster\\
    When an instance terminates and a large cluster $C$ is formed, we mark all free processes in $C$. 
    The cost of merging the pieces of component $C$ into a single cluster, which is at most 
    $|C|$, is then distributed equally among the marked free processes in $C$.
    \item Merging with a large cluster\\
    Suppose, the \texttt{CompMerge($A, B$)} operation is executed, where one of the components is a large cluster. Let $A$ be this component. The cost of this operation is at most $|B|$. 
    %Subsequently, the smaller component merges into the larger component.
    We unmark all marked processes in $B$ and distribute the cost of the merge equally among the marked free processes in $A$.
\end{itemize}

Now, we calculate the charge of free processes $v \in C$, where $C$ is a large cluster. Notice that only marked processes are charged. 
By construction, there are at least $\sqrt{k}$ marked free processes in $C$. Since $|C| \le k$, at most $k$ processes could have migrated into $C$ due to the above two cases. Therefore, each marked free process is charged at most $\frac{k}{\sqrt{k}} = \sqrt{k}$ until it is unmarked. Then, $\CLarge \le \sum_{v \in F} \sqrt{k} = \sqrt{k} |F|$.
\end{proof}

Next, we establish bounds on $\CSmall$ and $\CMono$, using the charging scheme presented in Appendix~\ref{appendix_c}.

\subsubsection*{Upper bound on $\CSmall$}
Assume, we merge two $\FLOW$ instances $I(C_A)$ and $I(C_B)$ into a combined instance $I(C_A \cup C_B)$. 
The merging process involves following operations: 
\begin{itemize}
    \item Merging special pieces\\ 
    The two special pieces (which by design are singleton clusters) merge into one special piece.
    \item Merging colored pieces\\
    Pieces with the same color $c$ merge into a single $c$-colored piece.
\end{itemize}
The cost of the instance merge operation is defined as the number of processes that change clusters. A merge of two pieces is implemented by migrating the processes of the smaller piece to the larger piece, subsequently deleting the smaller piece. Importantly, not every merge of pieces incurs a cost. 

Let $A \subseteq C_A$ and $B\subseteq C_B$ be two pieces (either special or colored) that merge. We distinguish between two types of merges between pieces:
\begin{itemize}
    \item Monochromatic merge\\
    $A$ and $B$ belong to the same colored cluster. In this case, no cost is incurred, since no process changes its cluster.
    \item Non-monochromatic merge\\
    Either $A$ or $B$ is a singleton cluster. The processes of the smaller piece migrate to the larger piece. The cost of this merge is $\min(|A|, |B|)$. 
\end{itemize}

Now, we analyze the costs of non-monochromatic merges. We develop a charging scheme to distribute these costs to the free processes in $C_A \cup C_B$. 
%Assume that a non-monochromatic merge is executed, i.e., two pieces $A$ and $B$ merge, where one of the pieces is a singleton cluster.
Without loss of generality, assume that $A$ is a singleton cluster. By construction, we have that $\ffv(A) \ge \frac{1}{4}$.
For a free process $v$ we distinguish between two types of charges: $\chf(v)$ and $\chm(v)$.
%Furthermore, we differentiate between two cases:
\begin{itemize}
    \item Case 1: $\ffv(B) \ge \frac{1}{16}$:\\
    W.l.o.g., let $|C_A| \le |C_B|$. The cost of this merge operation is $\min(|A|,|B|) \le 16 \min(|F_A|,|F_B|) \le 16 |F_A|$. We distribute this cost by increasing $\chf(v)$ by $16$ for each $v \in F_A$. 
    \item Case 2: $\ffv(B) < \frac{1}{16}$:\\
    We charge each $v \in F_A \cup F_B$ the cost $\frac{\min(|A|,|B|)}{|F_A|+|F_B|}$.
\end{itemize}

Note, that each free process is charged at most once during an instance merge. The above two cases imply that the total cost $\CSmall$ is bounded by:
\[
\CSmall \le \sum_{v \in F} (\chf(v) + \chm(v))\enspace.
\]

%Now, we analyze  $\chf(v)$ and $\chm(v)$ in detail.
Next, we will analyze the detailed contributions of $\chf(v)$ and $\chm(v)$ to the overall cost.
\begin{lemma}
For each free process $v$, $\chf(v) \le 16\log k$.
\label{lemma:chf}
\end{lemma}
\begin{proof}
Fix a process $v$, and assume that $\chf(v)$ increases during an instance merge. 
This increase happens only when $v$ is part of the smaller connected component, i.e., $v \in C_A$ with $|C_A| \le |C_B|$.

After such a merge, $v$ belongs to the new connected component $C_A \cup C_B$ with $|C_A \cup C_B| \ge 2 |C_A|$. This means that the size of the connected component containing 
$v$ at least doubles each time $\chf(v)$ increases. 

Since the size of any component is bounded by $k$, it can only double at most $\log k$ times before it reaches $k$. Therefore, $\chf(v)$ can increase at most $\log k$ times.
Given that each increase is by $16$, we conclude that $\chf(v) \le 16\log k$ for each free process $v$.
\end{proof}

Now, we derive the contribution of $\chm(v)$.
\begin{lemma}
$\sum_{v \in V} \chm(v) \le 64\sqrt{k} |F|$
\label{lemma:chm}
\end{lemma}
\begin{proof}
Consider an iteration of a $\FLOW$ procedure. 
It ends only when a new flow-augmenting path is created by the procedure. 
Importantly, each colored cluster can only grow during an iteration.
%Consequently, merging two procedures does not end an iteration.

Assume, that at the beginning of an iteration $\chm(v) = 0$ for each $v \in C$ holds. Let $A$ be a piece maintained by the procedure. 
We utilize the charging scheme presented in Appendix~\ref{appendix_c} and define
\[
\avg(A) = \frac{\sum_{v \in F_A}\chm(v)}{F_A}\enspace,
\]
as the average charge of the free processes in $A$. 

We maintain the following invariant: For each piece $A$ in the flow procedure, during the current iteration, the inequality $\avg(A) \le 64(1-\ffv(A))$ holds. 

At the beginning of an iteration, $\avg(A) = 0$ for each colored cluster $A$, so the invariant is trivially satisfied. 
Additionally, each free process $v$ in the special cluster is interpreted as a singleton piece
$\{v\}$ with $\avg(\{v\})=0$, satisfying the base case. 
For the induction step, assume that the invariant holds prior to the current request.

There are several ways in which a piece $A$ can change:
\begin{itemize}
    \item 
    A free process is added to the piece $A$ by the $\FLOW$ procedure.\\
    Since $\chm(v)$ remains unchanged for all processes $v \in A \cup \{v\}$, we apply Lemma~\ref{lemma:charge_merge} to sets $A$ and $\{v\}$. Consequently, $\avg(A \cup \{v\}) \le 64(1-\ffv(A \cup \{v\}))$.
    \item An instance merge is executed, potentially merging $A$ with another piece $B$. We distinguish between two cases:
    \begin{itemize}
        \item $\chm(v)$ remains unchanged for all $v \in A \cup B$.\\
        We apply Lemma~\ref{lemma:charge_merge} to sets $A$ and $B$, ensuring $\avg(A \cup B) \le 64(1-\ffv(A \cup B))$.
        \item A charge is applied to $\chm(v)$ for some $v \in A \cup B$.\\
        This happens only when $\ffv(A) \ge \frac{1}{4}$ and $\ffv(B) \le \frac{1}{16}$ (or vice versa). In this case, we apply Lemma~\ref{lemma:charge_merge_cost} to sets $A$ and $B$. Hence, $\avg(A \cup B) \le 64(1-\ffv(A \cup B))$.
    \end{itemize}
\end{itemize}

From the above cases, we conclude that  $\avg(A) \le 64(1-\ffv(A)) \le 64$ throughout the entire iteration. 
Thus, each free process is charged at most $64$ during a single $\FLOW$ iteration, \emph{on average}. Since there are at most $\sqrt{k}$ iterations (then the procedure terminates), the total cost incurred across all iterations is:
\[
\sum_{v \in F} \chm(v) = \sum_{A} \sum_{v \in F_A} \chm(v) =  \sum_{A} \avg(A)|F_A| \le \sum_{A} 64 \sqrt{k} |F_A| = 64\sqrt{k}|F|\enspace.
\]
\end{proof}

Finally, we can state the following lemma:
\begin{lemma}
$\CSmall \le \mcO(\sqrt{k})|F|\enspace.$
\label{lemma:costMerge}
\end{lemma}
\begin{proof}
Applying Lemma~\ref{lemma:chf} and Lemma~\ref{lemma:chm} gives,
\begin{align*}
\CSmall 
&\le \sum_{v \in F} (\chf(v) + \chm(v))  \\
&= \sum_{v \in F} \chf(v) + \sum_{v \in F} \chm(v)\\
&\le 16\log k |F| + 64\sqrt{k} |F| \\
&\le 80\sqrt{k}|F|\enspace.
\end{align*}
\end{proof}

\subsubsection*{Upper bound on $\CMono$}
The following scenarios can increase $\CMono$: 
\begin{itemize}
    \item Case 1: At the end of an iteration of a $\FLOW$ procedure, all free processes migrate to the special piece, making all colored pieces 1-monochromatic. As a result, these colored pieces are 
    then migrated back to their corresponding colored cluster. 
    \item Case 2: A piece $A$ that was previously assigned to a colored cluster becomes $\gamma$-monochromatic with $\gamma < \frac{1}{2}$. This means, the piece $A$ now contains more than $\frac{1}{2}|A|$ free processes. Consequently, $A$ is removed from the colored cluster and becomes a singleton cluster. The cost of this operation is $|A|$.
    \item Case 3: A piece $A$ that was previously a singleton cluster becomes $\gamma$-monochromatic with $\gamma > \frac{3}{4}$. This implies that $A$ now contains fewer than $\frac{1}{4}$ free processes, leading to its migration back to its corresponding colored cluster. The cost of this operation is $|A|$.
\end{itemize}

We now demonstrate that $\CMono$ is bounded in terms of $\CSmall$ and $\CFlow$.
\begin{lemma}
$\CMono \le \mcO(\CSmall + \CFlow)\enspace.$
\label{lemma:costMono}
\end{lemma}
\begin{proof}
We analyze each case separately.
\begin{itemize}
    \item Case 1:\\By construction, only pieces $A$ with at least $\frac{1}{4}|A|$ free processes are not in their corresponding colored cluster at the end of an iteration. This also means, that the $\FLOW$ procedure moved at least $\frac{1}{4}|A|$ processes during the current iteration to piece $A$. Each such piece is moved to its corresponding colored cluster, with cost $|A|$. Hence, by a charging argument, the overall cost of Case 1 is at most $4\CFlow$. 
    \item Case 2 and 3
    \begin{itemize}
        \item First, we analyze Case 2. Consider a piece $A$ that is either initially within its colored server, or was moved there recently. 
        By construction, $\frac{|F_A|}{|A|} \le \frac{1}{4}$. Suppose, that after some instance merge operations, and/or free process migrations triggered by the $\FLOW$ procedure, the instance grows in size. 
        
        Let $A'$ be the expanded piece, and assume $|F_A'| > \frac{1}{2} |A'|$. 
        By construction, only up to $\frac{1}{4}|A'|$ of the free processes in $A'$ could have migrated as a result of monochromatic merges. Therefore, at least $\frac{1}{4}|A'|$ of the free processes in $A'$ must have migrated either due to non-monochromatic merges or through the actions of the $\FLOW$ procedure. Each of these migrations contributes either to $\CSmall$ or $\CFlow$. 
        
        Hence, the piece $A$ observed at least $\frac{1}{4}|A'|$ cost related to $\CSmall$ and/or $\CFlow$, since it was last placed in its colored cluster. On the other hand, the cost to convert $A'$ into a singleton cluster is $|A'|$. Thus, by applying a charging argument, the total cost for Case 2 is bounded by $4 (\CSmall + \CFlow)$.
        
        \item The proof for Case 3 follows a similar approach to that of Case 2, but focuses on non-free processes instead of free processes. 
        Consider a piece $A$ that recently became a singleton cluster. By construction, $|F_A|\ge \frac{1}{2}|A|$. Suppose, that after some instance merge operations, the instance grows, and let $A'$ denote the expanded piece. 
        
         Assume $|F_A'| < \frac{1}{4} |A'|$, indicating that $A'$ migrates to its corresponding colored cluster. 
         Since $A'$ became a singleton cluster, at least $\frac{1}{2}|A'|$ non-free processes must have migrated to $A'$ due to non-monochromatic merges. 
         
         Therefore, the piece $A$ incurred at least $\frac{1}{2}|A'|$ cost related to $\CSmall$, since it became a singleton cluster. Meanwhile, the cost to migrate $A'$ to the colored cluster is $|A'|$. Thus, by a charging argument, the overall cost for Case 3 is bounded by $2\CSmall$.
    \end{itemize}
    %Thus, by a charging argument, the costs for Case 2 and Case 3 are most $8 (\CSmall + \CFlow)$.
\end{itemize}
From the above cases, we conclude, that $\CMono$ is at most $\mcO(\CSmall + \CFlow)$.
\end{proof}

Finally, we are able to derive a bound on $\CostCluster$:
\begin{lemma}
$\CostCluster \le \mcO(\sqrt{k}) |\mathcal{P}|\enspace.$
\label{lemma:costCluster}
\end{lemma}
\begin{proof}
From Lemmas~\ref{lemma:costInst},\ref{lemma:costLarge},\ref{lemma:costMerge}, and \ref{lemma:costMono} we conclude that $\CostCluster \le \mcO(\sqrt{k}) |F|$. Since, by Lemma~\ref{lemma:flow_dual} and by construction of the instance/component merge operation, $|F|=2|\mathcal{P}|$, the result follows directly from these bounds.
\end{proof}

\subsection*{Scheduling Cost}
Next, we argue that the scheduling cost due to the scheduling algorithm is at most $\mcO(\CostCluster)$. The proof is almost entirely transferred from \cite{racke2023polylog}. For completeness, we present it again. 

\begin{lemma}[\cite{racke2023polylog}]
$\CostSchedule \le \mcO(\frac{1}{\epsilon}) \CostCluster$.
\label{lemma:costSchedule}
\end{lemma}
\begin{proof}
Initially, all servers are balanced, i.e., each server holds exactly $k$ processes. After the scheduling algorithm ends, all servers have at most $2k$ processes assigned. Now, if at the beginning of the scheduling algorithm, there is some server $s$ with at least $2+\epsilon$ processes, this means that since the last execution of the scheduling algorithm on this server, at least $\epsilon k$ processes migrated there due to the clustering algorithm. Executing the scheduling algorithm on $s$ costs at most $\mcO(k)$. Thus, by a charging argument, the overall cost of the scheduling algorithm is at most $\mcO(\frac{1}{\epsilon}) \CostCluster$. 
\end{proof}

\subsection*{Overall Cost}
\begin{lemma}
The cost of the algorithm is at most $\mcO(\sqrt{k}) \OPT$. 
\label{lemma:cost_two_eps}
\end{lemma}
\begin{proof}
%Let $\mathcal{P}$ be the set of free nodes at the current time step.
According to Lemmas~\ref{lemma:costCluster} and \ref{lemma:costSchedule}, the overall cost of the algorithm is at most 
\[
\CostCluster + \CostSchedule \le \mcO(\tfrac{1}{\epsilon}\sqrt{k}) |\mathcal{P}|. 
\]
According to Lemma~\ref{lemma:edge_disjoint_all}, the flow-augmenting paths in $\mathcal{P}$ are edge disjoint. 
Applying Lemma~\ref{lemma:edge_disjoint} yields the result.
\end{proof}

Applying Lemmas \ref{lemma:cost_two_eps} and \ref{lemma:feas_two_eps} we can finally state the theorem:
\UpperBoundsTwoEps*

%%% Local Variables:
%%% mode: latex
%%% TeX-master: "itcs24.tex"
%%% End:
\clearpage
\appendix

\section{Omitted Proofs}\label{appendix}

\subsection{Proof of Lemma~\ref{lemma:soda_cost}}\label{appendix_a}
\LemmaSodaCost*
\begin{proof}
Although not explicitly stated in \cite{soda21}, the lemma can be easily derived using the charging scheme outlined there.
We explain briefly, how the cost of the $\SODA$ algorithm in \cite{soda21} is allocated. 
Specifically, the cost is either assigned to processes, denoted by $\operatorname{charge}(v)$ for each process $v$, or as an extra charge $\operatorname{Extra}$.

We adapt the lemmas from \cite{soda21} to align with our notation:
\begin{itemize}
    \item Lemma 13: The total extra charge $\operatorname{Extra}$ is at most $\mcO(k \ell \log k)$.
    \item Lemma 14: The maximum process charge that a process $v$ can receive is at most $\mcO(\log k)$.
\end{itemize}
From these two lemmas, we conclude that the total cost of the $\SODA$ algorithm is bounded by 
\[
\SODA = \sum_{v} \operatorname{charge}(v) + \operatorname{Extra} \le \sum_{v} c \log k + c' k \ell \log k \le (c+c')k \ell \log k,
\]
for some constants $c$ and $c'$.
\end{proof}

\subsection{Charging Scheme for Lemma~\ref{lemma:costMerge}}\label{appendix_c}

Consider a set of processes $A$, with a subset $F_A \subseteq A$ representing the free processes in $A$. We define a charging scheme to distribute a charge among these free processes.

%\paragraph{Charge definition}
For each free process $v \in A$, let 
$\charge(v)$ denote the charge assigned to $v$.
The total charge for the set $A$ is the sum of the charges of all free processes in $A$:
\[
\charge(A) = \sum_{v \in F_A} \charge(v)\enspace.
\]

We then define the average charge distributed among the free processes in the set $A$ as
\[
\avg(A) = \frac{\charge(A)}{|F_A|} = \frac{ \sum_{v \in F_A} \charge(v)}{|F_A|}\enspace.
\]

Next, we introduce the free process fraction of a set $A$, denoted by $\ffv(A)$: 
\[
\ffv(A) = \frac{|F_A|}{|A|}\enspace.
\]

\noindent
Furthermore, for two disjoint sets $A$ and $B$, we define the following weighted averages:
\begin{equation*}
\savg_{A,B}(X,Y) = \tfrac{|A|}{|A|+|B|}X+\tfrac{|B|}{|A|+|B|}Y\text{~~~and~~~}
\favg_{A,B}(X,Y) = \tfrac{|F_A|}{|F_A|+|F_B|}X+\tfrac{|F_B|}{|F_A|+|F_B|}Y\enspace.
\end{equation*}
For simplicity, we may omit the subscripts $A,B$ when the sets are clear from the context.
Clearly, we have $\min\{X,Y\}\le\savg(X,Y)\le\max\{X,Y\}$ and
$\min\{X,Y\}\le\favg(X,Y)\le\max\{X,Y\}$ as this holds for any weighted average.
The following claim essentially states that for $\ffv(A)\le\ffv(B)$ the
function $\favg$ tends more towards the second parameter (when compared to $\savg$).
\begin{fact}
Let $X\ge Y$ and $\ffv(A)\le\ffv(B)$. Then
\begin{equation*}
\favg(X,Y) \le \savg(X,Y)\enspace.
\end{equation*}
\end{fact}
\begin{proof}
Both functions are weighted averages, i.e., functions of the form
$(1-\lambda)X+\lambda Y$. Therefore, it is sufficient to show that the
multiplier for the second parameter is larger in the function $\favg$. Observe
that
\begin{equation*}
\lambda_{\operatorname{free}}
=\tfrac{|F_B|}{|F_A|+|F_B|}
=\tfrac{|B|\cdot\ffv(B)}{|A|\cdot\ffv(A)+|B|\cdot\ffv(B)}
\ge\tfrac{|B|\cdot\ffv(B)}{|A|\cdot\ffv(B)+|B|\cdot\ffv(B)}
=\tfrac{|B|}{|A|+|B|}=\lambda_{\operatorname{size}}\enspace,
\end{equation*}
where the inequality holds because $\ffv(A)\le\ffv(B)$.
\end{proof}

\begin{fact}
Suppose that we merge two sets $A$ and $B$ into $A\cup B$. Then
\begin{equation*}
\avg(A\cup B) = \favg(\avg(A),\avg(B)) \text{~~~and~~~} \ffv(A\cup B) = \savg(\ffv(A),\ffv(B))\enspace.
\end{equation*}
\end{fact}
\begin{proof}
The total charge in $A\cup B$ is $\avg(A)\cdot |F_A|+\avg(B)\cdot |F_B|$.
Dividing this by the number $|F_A|+|F_B|$ of free processes in $A\cup B$ gives
the first equality. The total number of free processes in $A\cup B$ is
$\ffv(A)\cdot |A|+\ffv(B)\cdot |B|$. Dividing this by the total number $|A|+|B|$ of
processes in $A\cup B$ gives the second equality.
\end{proof}
\begin{lemma}\label{lemma:charge_merge}
Suppose we are given two sets $A, B$ that fulfill
\begin{equation*}
\avg(A) \le c(1-\ffv(A))\text{~~~and~~~}\avg(B) \le c(1-\ffv(B))
\end{equation*}
for a constant $c > 0$.
Then, $\avg(A\cup B)\le c(1-\ffv(A\cup B))$.
\end{lemma}
\begin{proof}
W.l.o.g.\ assume that $\ffv(A)\le\ffv(B)$. We differentiate between two cases:
\begin{itemize}
\item $\avg(A)\le\avg(B)$
\begin{align*}
  \avg(A\cup B)&=   \favg(\avg(A),\avg(B))\\
                  &\le \avg(B)\\
                  &\le c(1-\ffv(B))\\
                  &\le c(1-\savg(\ffv(A),\ffv(B)))\\
                  &=   c(1-\ffv(A\cup B))
\end{align*}
\item $\avg(A)\ge\avg(B)$
\begin{align*}
  \avg(A\cup B)&=   \favg(\avg(A),\avg(B))\\
                  &\le \savg(\avg(A),\avg(B))\\
                  &\le \savg(c(1-\ffv(A)),c(1-\ffv(B)))\\
                  &=c(1-\savg(\ffv(A),\ffv(B)))\\
                  &= c(1-\ffv(A\cup B))
\end{align*}
\end{itemize}
\end{proof}

\begin{lemma}\label{lemma:charge_merge_cost}
Suppose we are given two sets $A, B$ with  $\ffv(A) \ge \frac{1}{4}$ and $\ffv(B) \le \frac{1}{16}$ that fulfill
\begin{equation*}
\avg(A) \le 64(1-\ffv(A))\text{~~~and~~~}\avg(B) \le 64(1-\ffv(B))\enspace.
\end{equation*}
Assume a merge of $A$ and $B$ is executed, where the free processes $F_A \cup F_B$ are charged with cost $\min(|A|, |B|)$.
Then, after the merge, $\avg(A\cup B)\le 64(1-\ffv(A\cup B))$.
\end{lemma}
\begin{proof}
\begin{align*}
    \avg(A\cup B)     &=\favg(\avg(A),\avg(B)) + \frac{\min(|A|, |B|)}{|F_A| + |F_B|}\\
                         &\le \favg(64(1-\ffv(A)),64(1-\ffv(B))) + \frac{\min(|A|, |B|)}{|F_A| + |F_B|}\\
                         &= 64(1 - \favg(\ffv(A),\ffv(B))) + \frac{\min(|A|, |B|)}{|F_A| + |F_B|} \\
                         % &\le 64( 1 - \left( \ffv(A)\frac{|F_A|}{|F_A| + |F_B|} + \ffv(B)\frac{|F_B|}{|F_A| + |F_B|}\right) + 4 
\end{align*}
Furthermore, 
\begin{align*}
   \favg(\ffv(A),\ffv(B)) - \ffv(A \cup B) &= \frac{|F_A|^2}{|A|(|F_A| + |F_B|)} + \frac{|F_B|^2}{|B|(|F_A| + |F_B|)} - \frac{|F_A| + |F_B|}{|A| + |B|}\\
                            &= \frac{|F_A|^2|B| + |F_B|^2|A|}{|A||B|(|F_A| + |F_B|)} - \frac{|F_A| + |F_B|}{|A| + |B|}\\
                            &=\frac{(|B||F_A| - |A||F_B|)^2}{|A||B|(|A| + |B|)(|F_A|+|F_B|)} \\
\end{align*}

Let $g(x) = |B||F_A| - |A|x$. 
%Clearly, $g(0) = |B||F_A|$. 
For $0 \le x < \frac{|F_A|}{|A|} |B|$, the function is positive and monotonically decreasing. Hence, for $x \le \frac{1}{16}|B| < \frac{|F_A|} {|A|} |B| $ the function attains its minimum at $x= \frac{1}{16}|B|$. Therefore, $g(\frac{1}{16}|B|) = |B||F_A| - \frac{1}{16}|A||B| \ge \frac{1}{4}|A||B| - \frac{1}{16}|A||B| = \frac{3}{16}|A||B| > 0$. Thus, $(g(x))^2 \ge \frac{9}{16^2}|A|^2|B|^2 > \frac{1}{32}|A|^2|B|^2$, for $x \le \frac{1}{16}|B|$. Hence,

\begin{align*}
\frac{(|B||F_A| - |A||F_B|)^2}{|A||B|(|A| + |B|)(|F_A|+|F_B|)} 
&\ge \frac{1}{32}\frac{|A|^2|B|^2}{|A||B|(|A| + |B|)(|F_A|+|F_B|)} \\
&= \frac{1}{32}\frac{|A||B|}{(|A| + |B|)(|F_A|+|F_B|)} \\
&\ge \frac{1}{64}\frac{\min(|A|,|B|)\cdot \max(|A||B|)}{\max(|A||B|)\cdot(|F_A|+|F_B|)} \\
&= \frac{1}{64}\frac{\min(|A|,|B|)}{|F_A|+|F_B|} \enspace.
\end{align*}
The second inequality holds because $|A| + |B| \le 2 \max(|A|, |B|)$. 
%Therefore, $\frac{\min(|A|,|B|)}{|F_A| + |F_B|} \le 64 ( \favg(\ffv(A),\ffv(B)) - \ffv(A \cup B))$. 
Therefore,
\begin{align*}
\avg(A\cup B) 
&\le 64(1 - \favg(\ffv(A),\ffv(B))) + \frac{\min(|A|,|B|)}{|F_A| + |F_B|}\\
&\le 64(1-\favg(\ffv(A),\ffv(B))) + 64 ( \favg(\ffv(A),\ffv(B)) - \ffv(A \cup B)) \\  
&=  64(1-\ffv(A\cup B)).
\end{align*}
\end{proof}

\section{Overview for Deterministic Online Algorithms}
The following table provides a summary of the current results in the context of the online balanced partitioning problem. The results are categorized across three different models: the General Model (GM), the Learning Model (LM), and the Generalized Learning Model (GLM). 

\begin{table}[h]
\begin{tabular}{|c|c|c|c|}
\hline
\textbf{Model}                                        & \textbf{Augmentation} & \textbf{Upper Bounds} & \textbf{Lower Bounds} \\ \hline \hline
\multirow{3}{*}{\textbf{GM}}          & $2 + \epsilon$  & $\mcO(k \log k)$      & $\Omega(k)$          \\ \cline{2-4} 
                                                 & $1 + \epsilon$  & $\mcO(k \ell \log k)$ & $\Omega(\max(k, \ell \log k))$ \\ \cline{2-4} 
                                                 & $1$             & $\mcO((k\ell)^2)$           & $\Omega(k\ell)$            \\ \hline \hline
\multirow{3}{*}{\textbf{LM}}         & $2 + \epsilon$  & $\mcO(\log k)$        & $\Omega(\log k)$      \\ \cline{2-4} 
                                                 & $1 + \epsilon$  & $\mcO(\ell\log k)$    & $\Omega(\ell\log k)$  \\ \cline{2-4} 
                                                 & $1$             & $\mcO(k\ell)$             & $\Omega(k\ell)$            \\ \hline \hline
\multirow{3}{*}{\textbf{GLM}} & $2 + \epsilon$  & $\mcO(\sqrt{k})$      & $\Omega(\sqrt{k})$    \\ \cline{2-4} 
                                                 & $1 + \epsilon$  & $\mcO(\max(\sqrt{k\ell \log k}, \ell \log k)$                      & $\Omega(\max(\sqrt{k\ell \log k}, \ell \log k)$                      \\ \cline{2-4} 
                                                 & $1$             & $\mcO((k\ell)^2)$           & $\Omega((k\ell))$           \\ \hline
\end{tabular}
\end{table}
%%% Local Variables:
%%% mode: latex
%%% TeX-master: "itcs24.tex"
%%% End:

\bibliographystyle{plain}
\bibliography{bibliography}

\end{document}